\documentclass [amssymb, amsmath, pre,showpacs,nofootinbib]{revtex4}

\pdfoutput=1

\usepackage{color}
\usepackage{graphicx}% Include figure files
\usepackage{dcolumn}% Align table columns on decimal point
\usepackage{bm}% bold math
\usepackage[all]{xy}
\usepackage{subfigure}

\newtheorem{theorem}{Theorem}[section]
\newtheorem{lemma}[theorem]{Lemma}
\newtheorem{proposition}[theorem]{Proposition}
\newtheorem{corollary}[theorem]{Corollary}
\newtheorem{definition}[theorem]{Definition}

\newenvironment{proof}[1][Proof:]{\begin{trivlist}
\item[\hskip \labelsep {\bfseries #1}]}{\end{trivlist}}

\newenvironment{remark}[1][Remark(s):]{\begin{trivlist}
\item[\hskip \labelsep {\bfseries #1}]}{\end{trivlist}}

\newcommand{\qed}{\nobreak \ifvmode \relax \else
      \ifdim\lastskip<1.5em \hskip-\lastskip
      \hskip1.5em plus0em minus0.5em \fi \nobreak
      \vrule height0.75em width0.5em depth0.25em\fi}

\begin{document}

\title{Planar Graphical Models which are Easy}
\author {Vladimir Y. \surname{Chernyak}$^{a,b}$}
\author{Michael \surname{Chertkov}$^{b,c}$}

\affiliation{$^a$Department of Chemistry, Wayne State University,
5101 Cass Ave,Detroit, MI 48202\\
$^b$Center for Nonlinear Studies and Theoretical Division, LANL, Los Alamos, NM 87545\\
$^c$New Mexico Consortium, Los Alamos, NM 87544, USA}

\date{\today}

\begin{abstract}
We describe a rich family of binary variables statistical mechanics models on a given planar graph which are equivalent to Gaussian Grassmann Graphical models (free fermions) defined on the same graph. Calculation of the partition function (weighted counting) for such a model is easy (of polynomial complexity) as reducible to evaluation of a Pfaffian of a matrix of size equal to twice the number of edges in the graph. In particular,  this approach touches upon Holographic Algorithms of Valiant \cite{02Val,04Val,08Val} and utilizes the Gauge Transformations discussed in our previous works \cite{06CCa,06CCb,08CCT,08CCa,08CCb}.
\end{abstract}

\pacs{02.50.Tt, 64.60.Cn, 05.50.+q}

\maketitle
\section{Introduction}

This paper rests on classical results derived and discussed in statistical physics and computer science. Onsager solution \cite{44Ons} of the two-dimensional Ising model on a square grid, its combinatorial interpretation by Kac, Ward \cite{52KW} and Vdovichenko \cite{65Vdo}, and the relation between the Ising model and the dimer model, established independently by Temperley, Fisher \cite{61TF,66Fis} and Kasteleyn \cite{61Kas,63Kas}, have aimed mainly at analysis of phase transitions in parametrically homogeneous infinite systems. However, the algebraic and combinatorial techniques used in these papers were intrinsically microscopic, and thus suitable for analysis of a much broader class of problems,  e.g. these that are parametrically inhomogeneous (glassy) and are formulated on finite planar graphs. In fact, the approach was extended \cite{66Fis,67Kas} to inhomogeneous versions of these models formulated on an arbitrary finite planar graph of size $N$. It was shown that, under these conditions, the partition (generating) function of an arbitrary Ising model (however without magnetic field) or of
the dimer model are represented by Pfaffians (or determinants) of the properly defined $N\times N$ matrices. Note that computation of a Pfaffian/determinant requires utmost $N^3$ steps, while calculation of the partition function for
a generic statistical physics model is a task of likely exponential complexity. More accurately, this is the task of the $\#P$ complexity class according to computer science classification, started with the classical papers of Cook \cite{71Coo} and Karp \cite{72Kar}. The  $\#P$ feature of the generic statistical physics model makes the easiness of the planar Ising and dimer models surprising and exceptional,  especially in view of the results of Barahona  \cite{82Bar}, who showed that adding magnetic field (even homogeneous) to the planar Ising model immediately elevates the
weighted counting problem (this is the partition function computation) to be of the $\#P$-complete class, i.e. of the complexity necessary and sufficient for solving any other problem that belongs to the $\#P$ hierarchy \footnote{The  $\#P$, pronounced sharp-P, and $\#P$-complete classes were introduced by Valiant in \cite{79Val}, who studied the complexity of calculating the matrix permanent.}. Similarly,  adding monomers turns the dimer-monomer model on a planar graph (or even on a planar bi-partite graph) into a problem that belongs to the $\#P$-complete class \cite{97Vad}\footnote{Dimer and monomer-dimer models from statistical physics are called perfect matching and matching models, respectively, in computer
science.}.

One wonders if the easiness of the dimer and Ising models on planar graphs is a lucky exception or may be it is, in fact, just a little top of yet unexplored iceberg? A new approach of Valiant \cite{02Val,04Val,08Val}, coined by the author ``holographic algorithm", shed some additional light onto this question. In \cite{02Val,04Val,08Val} Valiant described a list of easy planar models reducible to dimer models on planar graphs via a set of ``gadgets". The gadgets were of ``classical" and ``holographic" types. A classical gadget describes an elementary graphical transformation, applied to the variables defined on a graph element, e.g. edge or a region, that preserves a one-to-one correspondence between the configurations of the original and transformed models. An example of a classical gadget would be an approach used by Fisher \cite{66Fis} (point-to-triangle transformation) to map the Ising model onto a dimer model. A holographic gadget of \cite{02Val,04Val,08Val} involves a linear transformation of the parametrization basis for the binary variables, so that the solution fragments of the original and mapped models would be in a mixed relation when certain sums, rather than individual elements of the original and derived formulations, are related to each other. The freedom in choosing an arbitrary nonsingular basis for the holographic transformation was
discussed in \cite{02Val,04Val,08Val}, however, it was explored there in a somewhat limited fashion. One emphasis of \cite{02Val,04Val,08Val} was on generating a polynomial time algorithm (via reduction to a determinant) for a number of problems for which only exponential time algorithms where known before, such as counting the edge orientations on a planar graph of maximal degree $3$ with no nodes containing all the edges directed towards or away from it. This model belongs to the class of ``ice" problems studied earlier in statistical physics \cite{67Lie_a,67Lie_b,07Bax}. We will actually discuss this example in details later in Section \ref{subsec:ice}. A class of easy factor-function models, stated in terms of binary edge variables on planar graphs of degree not larger than three, was also discussed in \cite{08CCT}. These models were stated in terms of a set of algebraic equations, one per vertex, constraining the factor-functions of the model.

In this paper we present a family of statistical models, stated in terms of binary variables defined on the edges of a given planar graph, with the partition functions reducible to Pfaffians of square matrices of the size equal to twice the number of graph edges. The general model of this class (easy on a given planar graph) can be defined in two consecutive steps: (a) construct an arbitrary Wick Binary Graphical Model (WBG) on the graph, and (b) further apply an arbitrary Gauge transformation of the type discussed in our early work devoted to the Loop Calculus (LC) approach \cite{06CCa,06CCb,08CCT,08CCa,08CCb}. Of the two steps, both required to characterize this class of easy binary models on a given graph, the description of the easy WBG models on a given planar graph and consecutively showing that it is Pfaffian (or det)-easy,  is the most novel contribution of the paper. Therefore, we find it useful to informally describe a general WBG model already here in the Introduction. (See Section \ref{subsec:WBG_def} for the formal description.)

A WBG model is defined in terms of painting the graph edges (each edge can be colored or not) with the ``even coloring" requirement (for each vertex the number of colored edges adjusted to it should be even). The weight of an allowed configuration (an even coloring) is given by the product of the vertex weights/contributions. A vertex contribution is equal to unity when no edges adjusted to the vertex are colored. The vertex weights are introduced for all pair colorings, defined as painting of pairs of edges adjusted to the vertex.  Thus, the total number of unconstrained parameters associated with a vertex is $N_v(N_v-1)/2$, where $N_v$ is the vertex degree (valence), i.e., the number of edges adjusted to the vertex. The weights of all higher-degree colorings ($4,6,\cdots$ edges of the vertex are colored) are not independent, but rather explicitly expressed in terms of the pair weights according to a simple rule, illustrated below using an example of a vertex of degree four. Assume that the edges of the degree four vertex are numbered clockwise $1,2,3,4$ with the corresponding weights of the pair-wise colorings being $w_{12}, w_{13},w_{14},w_{23},w_{24},w_{34}$.  Then, the weight of the four-edge coloring is $w_{12}w_{34}-w_{13}w_{24}+w_{14}w_{23}$. We refer to the described rule and the corresponding model as the Wick rule/model, since modulo signs, the expression for the higher-order weights is represented by the Wick's decomposition of a higher-order (even) correlation functions of a Gaussian random field in terms of the pair correlation functions (covariances). The sign rule for an individual contribution (that corresponds to a partitioning the even number of $2k$ colored edges into $k$ pairs) to the weight is defined as follows: (i) represent the edges by points on a circle (according to their cycling ordering),
(ii) connect the points that represent the paired edges with lines, (iii) count the total number of the line crossings inside the circle, (iv) choose plus/minus sign if the number is even/odd. Note, that the Ising model, the dimer model and some other planar easy models discussed in \cite{02Val,04Val,08Val}, can be represented, after a gauge transformation, in terms of a special case of the general WBG model. For example, an WBG representation for the Ising model is directly evident from the respective high-temperature expansion/series. (See \cite{83Pop} and discussion of Section \ref{subsec:Ising}.)

In this manuscript we show that
\begin{itemize}
\item The WBG model with an arbitrary choice of the edge associated weights, $w$, is equivalent to corresponding Fermion Gaussian model, referred to as the Grassman Gaussian Graphical (G$^3$) model, and defined on the same graph. With a suitable  choice of the so-called {\it complete orientations} for edge and triplets of neighboring vertices contributing definition of the G$^3$ model (and thus called in the following {\it suitable complete orientation} of the graph, or when it is not confusing simply -{\it suitable orientation}), the two models are equivalent in the sense that their partition functions are equal to each other and also equal to an explicitly defined Pfaffian of a skew-symmetric matrix of the size twice the number of edges in the graph. Obviously, computational complexity of the object is polynomial in the graph size. This main result of the manuscript is formally stated in the Theorem \ref{Theorem}.
\item Application of an arbitrary and graph-local gauge transformation (described in terms of a pair of $2\times 2$ matrices per edge orthogonal to each other) to the WBG model generates another det- easy model.
\end{itemize}

This paper builds upon, and in a certain sense extends, the classical results of Kasteleyn \cite{63Kas,67Kas} and Temperley, Fisher \cite{61TF,66Fis} developed for a planar dimer model. One key ingredient of the Kastelyan approach was the construction of edge-orientations. It was shown that for a special choice of the edge-orientations, so-called Kasteleyn edge-orientations, the partition function of the dimer model becomes a Pfaffian of a skew symmetric matrix built from element-by-element product of the edge weights and the $\pm1$ Kasteleyn edge-orientation factors. As shown in \cite{63Kas,67Kas}, construction of the Kasteleyn orientation of a planar graph can be done efficiently in the number of steps linear in the size of the system. (See also Section 8.3 of \cite{86LP}.) Our scheme also requires defining graph-local orientations,  however of an extended character,  where not only edge-orientations,  but also triplet-orientations are needed. The combination of the edge-orientations and triplet-orientations defined on a given graph will be called {\it complete orientations}. A special choice of the {\it complete orientation}, the {\it suitable complete orientation}, in a sense generalizes the notion of the Kasteleyn edge-orientation. See Section \ref{sec:Def} for accurate definitions.

Another important graph-related structure utilized in this paper is the so-called
{\it even generalized loops}, decomposed into combinations of {\it embedded orbits} on the graph. For any finite graph these graph objects are finite, which allows to represent  relevant partition functions in terms of a sum over a finite number of contributions associated with embedded orbits. At this point it is also appropriate to refer to a related, yet different, approach that operates with (random) walks and associated immersed orbits, that has been  formulated for the Ising model in the classical papers of Kac, Ward \cite{52KW} and Vdovichenko \cite{65Vdo}. ( See also \cite{87Mor} built upon \cite{52KW,65Vdo} and describing how to use the Kac, Ward-Vdovichenko method to solve the dimer model.) In contrary with the case of the {\it embedded orbits}, which play a key role in our manuscript, the number of {\it immersed orbits} is infinite even on a finite graph (since the length of the immersed orbits is not restricted), and thus the full partition function for a finite graph (which is also equal to a finite determinant) is represented in terms of an infinite product (not sum!) over the equivalence classes of the walks over the cyclic permutations. In this context  the language of the Grassman/Fermion models extensively used in our mansucript is more appropriate for the finite graphical objects, i.e. embedded rather than more general immersed orbits, because of the nilpotent feature of the Grassmann variables.

We would also like to note that in this manuscript we did not give exhaustive description to the full class of  det-easy problems on a given planar graph but described its rather broad, but still incomplete, subset. This incompleteness is in fact illustrated on a ``counter-example"  of the so-called $\#X$-matching model of Valiant \cite{02Val,04Val,08Val}, which is reducible to a det-easy model, however, via a set of more general (extended alphabet) gauge transformations, rather than binary gauge transformations discussed in our manuscript.

The manuscript is organized as follows. Section \ref{sec:Def} starts with a brief description of the key graphical objects used through out the manuscript and then it is split into three Subsections.
Sections \ref{subsec:WBG_def} and \ref{subsec:G3_def} define the Wick Binary Model (WBG) model and the Grassmann Gaussian Graphical (G$^3$) model respectively,  while Section \ref{subsec:theorem} states the main result of the manuscript concerning relation between the two models under the proper, so-called {\it suitable},  choice of orientations entering definition of the G$^3$ model. The respective Theorem is proved in Section \ref{sec:WBG=G3} in a combinatorial way, while discussion of alternative topological considerations is given in Appendix \ref{sec:topology}. Section \ref{sec:suitable} is devoted to constructing efficiently a suitable orientation for a given planar graph. The construction of a suitable orientation is split into two steps.  First step, discussed in Section \ref{subsec:extended-induced}, describes a procedure of building an extended graph, defining a Kasteleyn orientation on it and then generating an {\it induced} orientation on the original graph. On the second step,  described in Section \ref{subsec:induced-suitable}, we show that the {\it induced} orientation is a {\it suitable} orientation, thus  resulting overall in a simple algorithm for evaluating the partition function of any  WBG model on the planar graph.
Section \ref{sec:gauge} discusses some details and examples of the gauge transformations, which transforms a WBG model to other det-easy models. The Section is broken in four Subsections. Section \ref{subsec:gauge-reminder}, based on materials from \cite{06CCa,06CCb,08CCT,08CCa,08CCb}, contains a brief review of the gauge transformation procedure, as well as examples that illustrate how the gauge transformations reduce the dimer model, ice model and Ising model on a planar graph to the general WBG models, given in Sections \ref{subsec:dimer-example},\ref{subsec:ice},\ref{subsec:Ising}
respectively. The aforementioned ``counter-example", which is det-easy yet requires a more general reduction procedure, is discussed in Section \ref{subsec:X-matching}. An alternative generalization of the gauge transformation approach is discussed in
Appendix \ref{sec:gauge-extend}. Section \ref{sec:concl} concludes the manuscript with a brief summary and a discussion of future challenges.

\section{Key Objects, Models and Statements}
\label{sec:Def}

We start describing briefly some common definitions used in this Section but also through out the paper.

Consider a finite connected non-oriented graph, ${\cal G}=({\cal G}_{0},{\cal G}_{1})$, with ${\cal G}_{0}$ and ${\cal G}_{1}\subset \{\alpha\in 2^{{\cal G}_{0}}|{\rm card}(\alpha)=2\}$ being the (finite) sets of graph nodes/vertices and (undirected) edges/links, respectively. For $a,b\in {\cal G}_{0}$ we will often use a notation $a\sim b$ to describe the relation $\{a,b\}\in {\cal G}_{1}$, i.e., ``$a$ and $b$ are connected by an edge''. A directed edge (that goes from node $a\in {\cal G}_{0}$ to node $b\in {\cal G}_{0}$) is an ordered pair $(a,b)\in{\cal G}_{0}\times {\cal G}_{0}$ with $\{a,b\}\in {\cal G}_{1}$. There are exactly two directed edges $(a,b)$ and $(b,a)$ that represent (are associated with) a non-directed edge $\{a,b\}\in {\cal G}_{1}$. A triplet $a\to b\to c$ is an ordered set $(a,b,c)\in {\cal G}_{0}\times {\cal G}_{0}\times {\cal G}_{0}$ with $\{a,b\},\{b,c\}\in {\cal G}_{1}$ and $a\ne c$. With a minimal abuse we will use a convenient notation $(a,b)\in {\cal G}_{1}$ or even $(a,b)\in {\cal G}$ to describe the relation $\{a,b\}\in {\cal G}_{1}$, and $(a\to b\to c)\in {\cal G}_{1}$ or $(a\to b\to c)\in {\cal G}$ to describe the relation ``$a\to b\to c$ is a triplet in ${\cal G}$''.

{\bf All graphs discussed in this manuscript are planar}, i.e., they can be drawn on the plane in such a way that their edges intersect only at the endpoints.

Let $\delta_{a}({\cal G})$ be the set of vertices which are neighbors of vertex $a$ over graph ${\cal G}$,  and $|\delta_{a}({\cal G})|$ be the cardinality of the set,  which we will also call degree/valence of $a$ in ${\cal G}$.
The {\it graph degree/valence} of ${\cal G}$ is the maximal $|\delta_{a}({\cal G})|$ among all the vertices of the graph. A graph is called {\it homogeneous} if all its vertices have the same valence. A connected graph ${\cal G}$ is called {\it irreducible} if after withdrawing any single edge the graph stays connected; it is called {\it reducible}, otherwise. Starting with a reducible graph and withdrawing one-by-one the edges that increase the number of connected components we arrive at a disjoint union of irreducible graphs, referred to as the {\it irreducible components} of ${\cal G}$. The edges withdrawn in the above procedure are referred to as the {\it connecting edges}. Note that if we contract the irreducible components of ${\cal G}$ to points, the resulting graph will become a tree, whose edges are represented by the connecting edges of ${\cal G}$. Also note that an irreducible graph does not have vertices of degree one. A {\it linear graph} ${\cal L}$ of length $l$ is a graph with the set of vertices and edges being ${\cal L}_0=\{a_{0},\ldots,a_{l}\}$ and ${\cal L}_1=\{\{a_{0},a_{1}\},\ldots,\{a_{l-1},a_{l}\}\}$ respectively. A {\it circular graph} ${\cal S}$ of length $l$ has the following set of vertices and edges: ${\cal S}_{0}=\{a_{1},\ldots,a_{l}\}$ and ${\cal S}_{1}=\{\{a_{1},a_{2}\},\ldots,\{a_{l-1},a_{l}\},\{a_{l},a_{1}\}\}$. For an irreducible graph ${\cal G}$ of degree three or higher a {\it linear segment} is a linear subgraph ${\cal L}\subset {\cal G}$ with the degree of vertices (in ${\cal G}$) $a_{1},\ldots,a_{l-1}$ being two, and degree of $a_{0},a_{l}$ exceeding two. For a  {\it circular segment} the vertices $a_{2},\ldots,a_{l}$ have degree two, whereas the degree of $a_{1}$ exceeds two. All irreducible graphs of degree two are circular. An irreducible graph of degree three does not have circular segments.

The layout of the the material in the remainder of the Section is as follows.
The Wick Binary Graphical (WBG) model on a planar graph and the auxiliary objects, e.g. even generalized loops and intersection index, required for the model definition are introduced in Section \ref{subsec:WBG_def}. Then in Section \ref{subsec:G3_def} we proceed to defining the Grassmann Gaussian Graphical (G$^3$) model and respective auxiliary graph objects, the complete (combined edge- and triplet-) orientation. Finally, in Section \ref{subsec:theorem} we define the notion of {\it suitable complete orientation}
and state the main result of the manuscript describing relation between the WBG and the G$^3$ models defined on the same graph.

\subsection{Wick Binary Graphical Model}
\label{subsec:WBG_def}

Consider a graph of an arbitrary degree and introduce the following notations
\begin{definition}[Generalized loop. Even Generalized loop$=\mathbb{Z}_{2}$-cycle.]
{\it Generalized loop} of ${\cal G}$ is a subgraph of ${\cal G}$ which does not contain
vertices of degree one (an empty set is also a generalized loop).
{\it Even generalized loop}, or $\mathbb{Z}_{2}$-cycle, is a subgraph with all vertices of even degree.
\end{definition}

\begin{remark}
Obviously, a simple cycle of ${\cal G}$ is also a generalized loop of ${\cal G}$. The term
$\mathbb{Z}_{2}$-cycle establishes a connection with a language common in algebraic topology.
A $\mathbb{Z}_{2}$-chain, $c\in {\cal C}_1({\cal G};\mathbb{Z}_{2})$,
is defined by a set of
$\mathbb{Z}_{2}$-numbers $c=\{c_{ab}=c_{ba}=0,1|{\{a,b\}\in{\cal G}_{1}}\}$.
We further introduce the standard boundary operator with its value at a vertex $a\in{\cal G}_0$:
$\partial({\bm c})_{a}=\sum_{b\sim a}c_{ab}  (\mbox{mod } 2)$. A chain $c$ is called a $\mathbb{Z}_{2}$-cycle, if it has no boundary $\partial({\bm c})=0$, i.e., $\forall a\in{\cal G}_0$:
$\partial({\bm c})_a=0$.
%Note that each configuration of a binary edge model is obviously represented by a $\mathbb{Z}_{2}$-chain from ${\cal C}_1({\cal G};\mathbb{Z}_{2})$. The requirement for and EBG model to be even means that the allowed configurations are represented by the $\mathbb{Z}_{2}$-cycles on the graph, i.e., cycles which belongs to ${\cal Z}_{1}({\cal G};\mathbb{Z}_{2})\equiv\{{\bm c}\in {\cal C}_{1}({\cal G};\mathbb{Z}_{2})|\partial({\bm c})=0\}$.
Even generalized loop, as a subgraph of ${\cal G}$ associated with a $\mathbb{Z}_2$ cycle $\gamma$, is constructed of edges with nonzero contributions, $\gamma_{ab}=1$. Therefore, in what follows we may
think of $\gamma$ as either a $\mathbb{Z}_2$ structure on the graph or as of respective subgraph.
\end{remark}

\begin{figure}
\centering \subfigure[]{\includegraphics[width=0.22\textwidth,page=4]{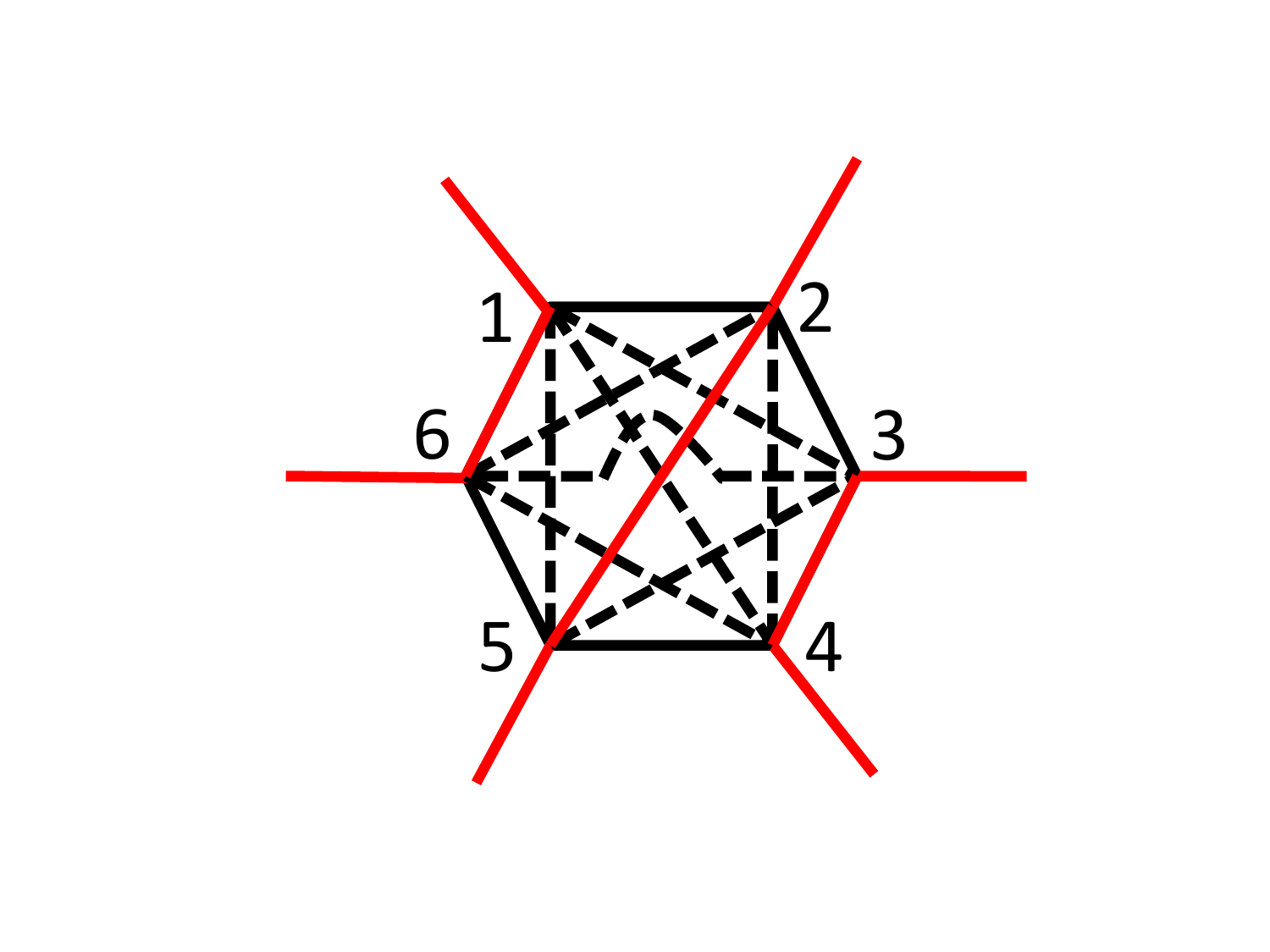}}
\subfigure[]{\includegraphics[width=0.22\textwidth,page=5]{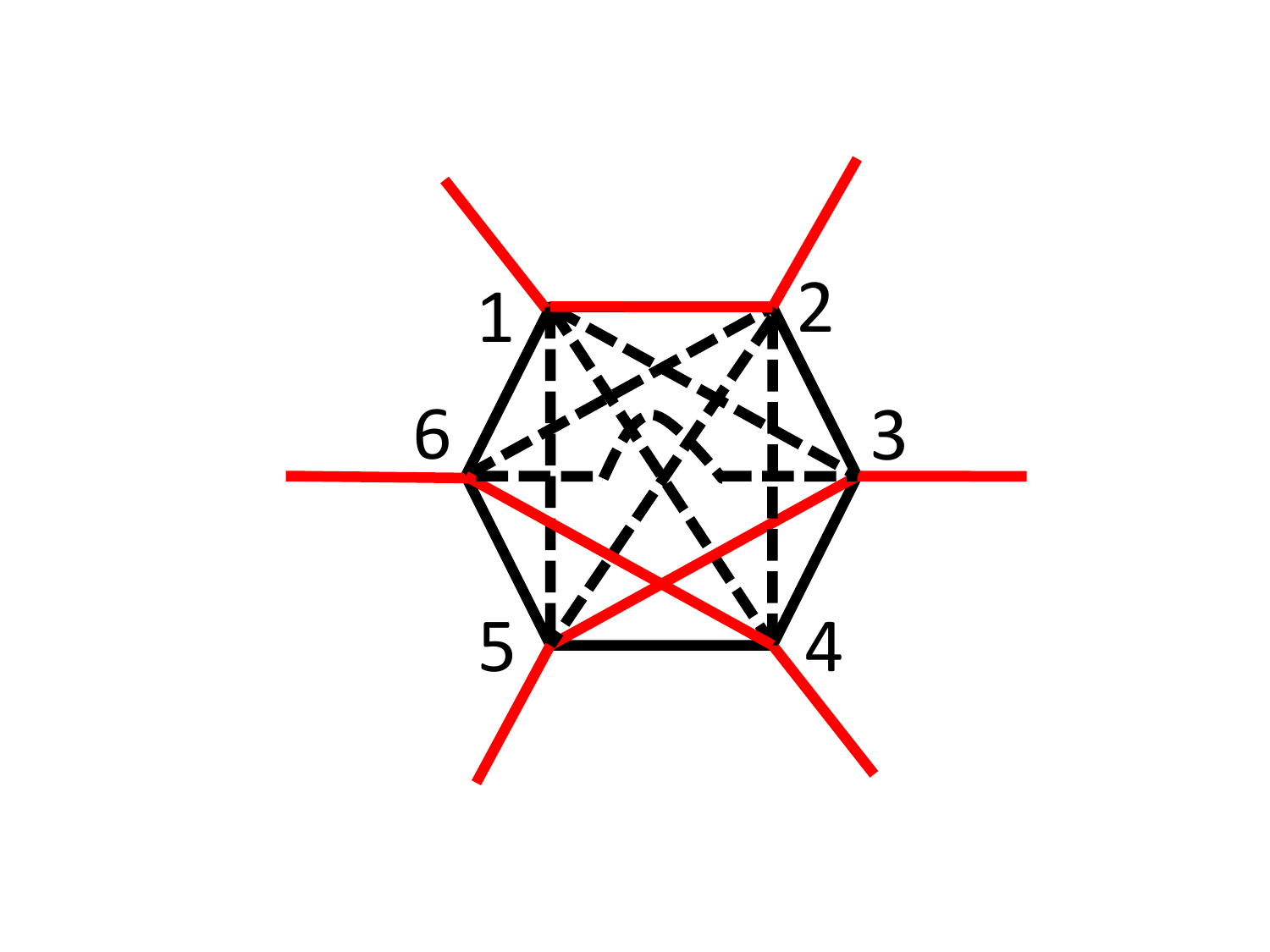}}
\subfigure[]{\includegraphics[width=0.22\textwidth,page=6]{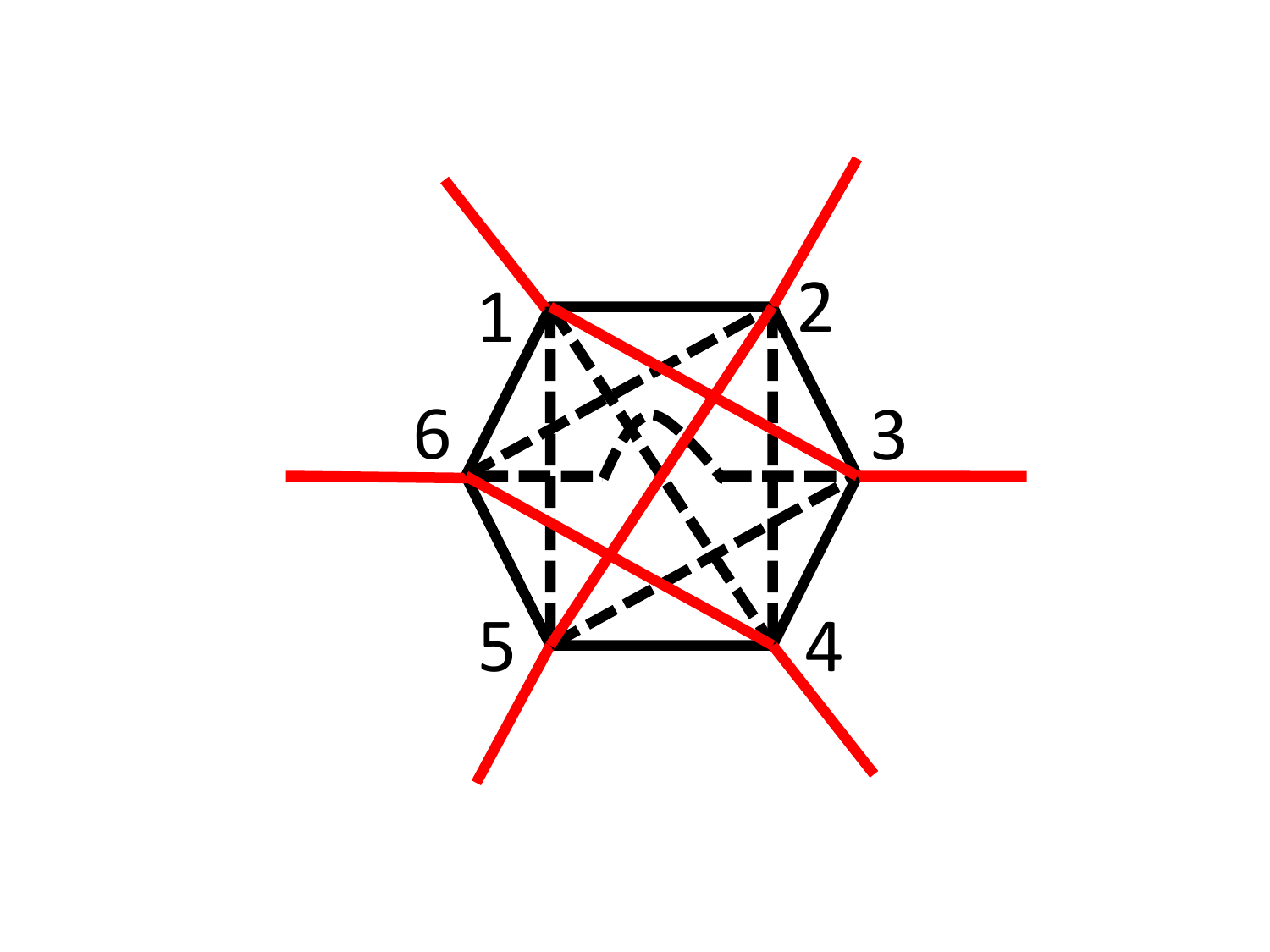}}
\subfigure[]{\includegraphics[width=0.22\textwidth,page=7]{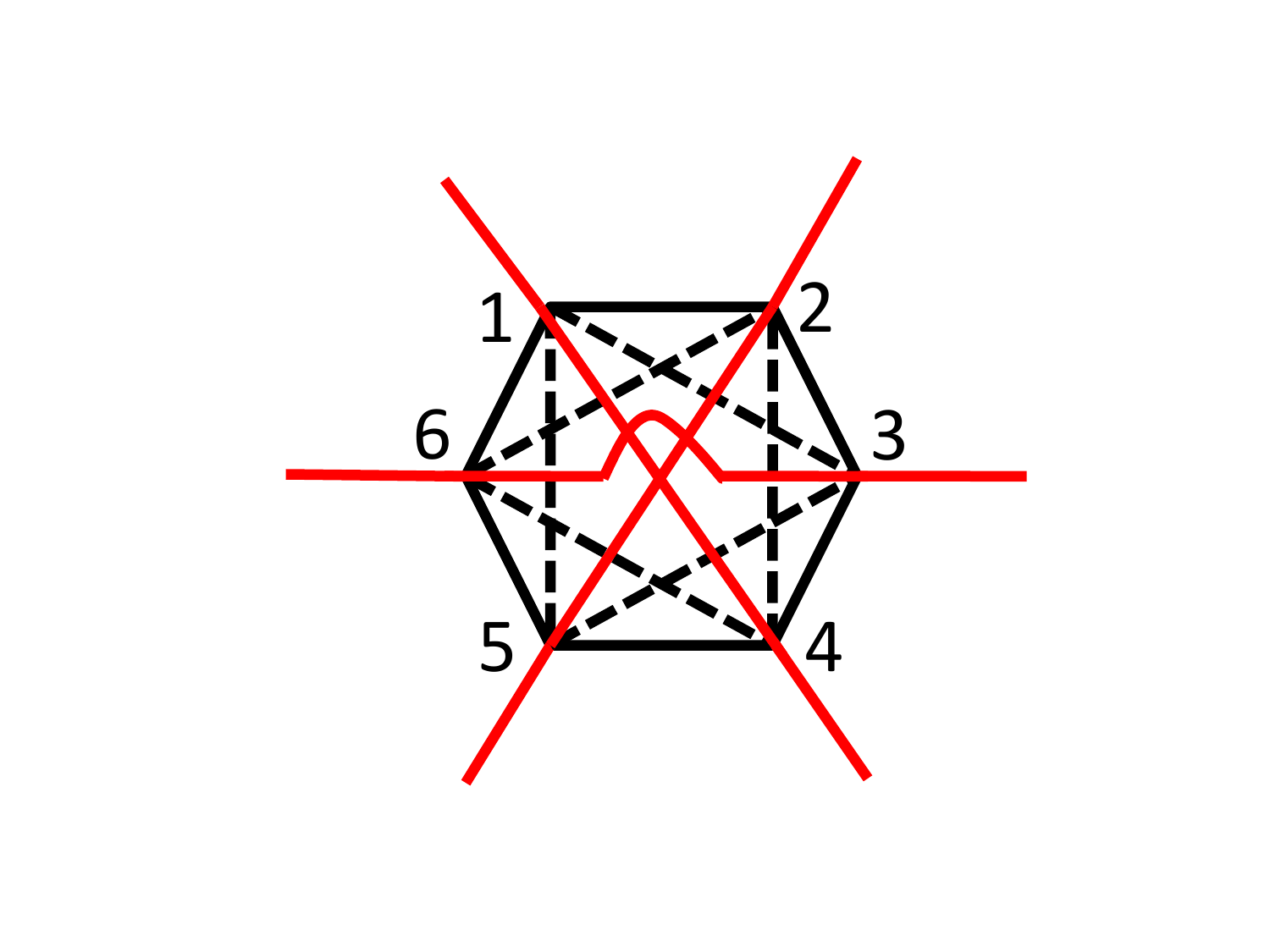}} \caption{Illustration of the
Wick model construction. Four graph illustrates four representative terms (of the total of
$(2n-1)!!$) for different $6$-vertex contributions into Eq.~(\ref{W-expressions}), correspondent to (a)
$W_{\{[1,6],[2,5],[3,4]\}}=W_{16}W_{25}W_{34}$ [zero crossing], (b) $W_{\{[1,2],[3,5],[4,6]\}}=-W_{12}W_{35}W_{46}$ [one crossing], (c)
$W_{\{[1,3],[2,5],[4,6]\}}=W_{13}W_{25}W_{46}$ [two crossings], and (d) $W_{\{[1,4],[2,5],[3,6]\}}=-W_{14}W_{25}W_{36}$ [three
crossings], where the upper index of the $6$-vertex is dropped to simplify the notations.\label{fig:6-Wick}}
\end{figure}

We will also need the following graph/vertex-local notations.

{\bf Local partitions and intersection index.} For a finite set $A$ with even number ${\rm card}(A)=2k$ of elements let $P(A)$ be the set of all partitions of $A$ in $k$ distinct non-ordered pairs. An element $\xi\in P(A)$ can be viewed as a set of $k$ elements, represented by pairs $\{a,a'\}\in\xi$ (with $a,a'\in A$) that form the partition. We further note that, since our graph ${\cal G}$ is planar, each vertex $b$ is equipped with a cyclic ordering (we choose counterclockwise) of its adjacent vertices (neighbors) $a\sim b$. For a vertex $b$ we can place its neighbors $a\sim b$ on a circle $S^{1}\in \mathbb{R}^{2}$ standardly embedded into a plane, according to their cyclic ordering, and denoting by $I_{\alpha}=I_{\{a,c\}}$ a segment of a straight line that connects $a$ and $c$ we introduce the modulo-$2$ intersection index $I_{\alpha}\cdot I_{\alpha'}\in \mathbb{Z}_{2}$, associated with two distinct pairs $\alpha=\{a,c\}$ and $\alpha'=\{a',c'\}$, as the number of intersections of $I_{\alpha}$ with $I_{\alpha'}$. Obviously we can associate with a partition $\xi\in P(\{a_{1},\ldots,a_{2k}\})$ of an even subset of neighbors of $b$ a $\mathbb{Z}_{2}$ intersection index by
\begin{eqnarray}
\label{intersect-index-node} N_{b}(\{a_{1},\ldots,a_{2k}\};\xi)=\sum_{\{\alpha,\alpha'\}\in 2^{\xi}}^{\alpha\ne \alpha'}I_{\alpha}\cdot I_{\alpha'} \; ({\rm mod}\; 2).
\end{eqnarray}

%Consider a vertex $b$ and introduce a cyclic ordering of an even number of edges adjusted to the vertex, %placing the respective set of neighboring vertices $a_1,\cdots,a_{2k}\sim b$ on a circle $S^{1}\in %\mathbb{R}^{2}$ standardly embedded into a plane. denote by $I_{aa'}$ a segment of a straight line that %connects $a$ and $a'$. Then, $I\cdot I'\in \mathbb{Z}_{2}$ is the modulo-two intersection index, which is %given by the number of intersections of $I$ and $I'$. $P([2k-1])$ is the set of all $2k$-th order %partitions of the ordered set $[2k-1]$ of $2k$ elements in distinct non-ordered pairs. $\xi$ is a set of %$k$ elements (pairs), $p=\{i,j\}$ where $a_i,a_j\sim b$, ordered in an arbitrary way, and %$\alpha(p)=\{a_i,a_j\}$ is our notation for the respective set of the pairs of vertices.
%\begin{definition}[Intersection Index. Partition of an even neigborhod of a vertex into pairs.]
%\label{def:int_ind}
%\end{definition}

We are now ready to define one of our basic models.
A binary statistical model on the graph is defined in terms of its partition function, describing a weighted sum over all allowed configurations.

{\bf WBG model}. The partition function of the Wick Binary Graphical (WBG) model, dependent on the vector of weights ${\bm W}$, is given by
\begin{eqnarray}
&& Z_{WBG}({\bm W})= \sum_{\gamma=\{\gamma_{ab}\}\in {\cal Z}_{1}({\cal G};{\mathbb Z}_{2})}\prod_{b\in{\cal G}_0}^{\exists a\sim b|\gamma_{ab}\neq 0}W^{(b)}_{\{a_1,\cdots, a_{2k}\}\equiv\{a|a\sim b;\gamma_{ab}=1\}},
\label{Z_WBG}\\
&&
\label{W-expressions} W_{\{a_{1},\cdots, a_{2k}\}}^{(b)}\equiv\sum_{\xi\in
P(\{a_{1},\ldots,a_{2k}\})} W_{\xi,a_1\cdots a_{2k}}^{(b)},\quad W_{\xi,a_1\cdots a_{2k}}^{(b)}\equiv
(-1)^{N_{b}(\{a_{1},\ldots,a_{2k}\};\xi)}\prod_{\alpha\in\xi}W_{\alpha}^{(b)},
\end{eqnarray}
where in the first line the summation goes over all $\mathbb{Z}_{2}$-cycles, i.e. even generalized loops (contributions associated with all odd generalized loops are assumed zero), of ${\cal G}$, and subsequent product is over all vertices of nonzero degree contributing the even generalized loop. The second line in Eq.~(\ref{W-expressions})  represents explicit expression of the higher vertex weight with $k\ge 2$ in terms of the lowest vertex weights $W_{a_{1}a_{2}}^{(b)}$, where the latter represents the case $k=1$.
%To describe the higher vertex weight we use the cyclic ordering of edges adjacent to the node $b$ to place all $a_1,\cdots,a_{2k}\sim b$ on a circle $S^{1}\in \mathbb{R}^{2}$ standardly embedded into a plane. We further denote by $I_{aa'}$ a segment of a straight (this is not necessary, yet convenient) line that connects $a$ and $a'$. For any such $I,I'$ denote $I\cdot I'\in \mathbb{Z}_{2}$ the modulo-two intersection index, which is given by the number of intersections of $I$ and $I'$.  Further, $P([2k-1])$ in Eq.~(\ref{W-expressions}) denotes the set of all partitions of an ordered set $[2k-1]$ of $2k$ elements in distinct non-ordered pairs, where $\xi$ is viewed as a set of $k$ elements (pairs), $p=\{i,j\}$ where $a_i,a_j\sim b$, ordered in an arbitrary way (the result is not sensitive to the ordering), and $\alpha(p)=\{a_i,a_j\}$.
%%%%%%%%%%%%%%%%%%%%%%%%%%%%%%%%%%%%%%%%%%%%%%%%%%%%%%%%%%%%%%%%%%%%%%%%%%%%%%%%%%%%%%%%%%%%%%%%%%%%%%
%\begin{definition}[WBG model.]
%\label{WBG}
%\end{definition}
%%%%%%%%%%%%%%%%%%%%%%%%%%%%%%%%%%%%%%%%%%%%%%%%%%%%%%%%%%%%%%%%%%%%%%%%%%%%%%%%%%%%%%%%%%%%%%%%%%%%%%
\begin{remark}
(a) Due to symmetric nature of the weights we can write $W_{\alpha}^{(b)}=W_{\{a_{i},a_{j}\}}^{(b)}=W_{a_{i}a_{j}}^{(b)}=W_{a_{j}a_{i}}^{(b)}$ for $\alpha=\{a_{i},a_{j}\}$.

\noindent (b) Overall, the expressions given by Eq.~(\ref{W-expressions}) can be interpreted as follows. Each contribution $W_{\xi,a_{1}\ldots a_{2k}}^{(b)}$ to the sum is determined by a partition of the set $\{a_{1},\ldots a_{2k}\}$ into distinct pairs $\alpha_{1},\ldots\alpha_{k}$. The contribution is given by the product $W_{\alpha_{1}}^{(b)}\ldots W_{\alpha_{k}}^{(b)}$ with the overall sign in
front to be plus or minus, depending whether the total number of intersections between the segments $I_{\alpha_{i}}$ is even or odd. To make the construction, required to define the Wick model, transparent we illustrate it for $|\delta_{\gamma}(a)|/2=3$ in Fig.~\ref{fig:6-Wick}.

\noindent (c) We call the model of (\ref{Z_WBG},\ref{W-expressions}) Wick model to emphasize relation to the Wick theorem/rules used in the quantum field theory to express high-order correlation function via second-moments (covariances) in the case of Gaussian statistics. More specifically, the $(2k)$-th order contribution $W_{\{a_{1},\cdots, a_{2k}\}}^{(b)}$ to the weight function is given by the Pfaffian of a $(2k)\times (2k)$ skew-symmetric matrix, with the matrix elements given by the binary weights $W_{a_{i}a_{j}}^{(b)}$, that can be viewed as a Gaussian integral over Grassman variables.

\noindent (d) The introduction of the set of constraints given by Eq.~(\ref{W-expressions}) is motivated by our desire to establish term-by-term relation between the yet to be defined  G$^3$ model and the discrete model (\ref{Z_WBG}).

\end{remark}

\subsection{Orientations and the Gaussian Grassmann Graphical Model}
\label{subsec:G3_def}
As in the preceding Subsection we will start introducing key objects (orientations) required to define the G$^3$ model.

\begin{definition}[Edge-orientation.]
\label{def:edge-orientation} A graph edge-orientation ${\bm\sigma}$ is a particular choice of a directed representative $(a,b)$ for any undirected edge $\{a,b\}\in {\cal G}_{1}$, or equivalently a set of numbers $\sigma_{ab}=\pm 1$ associated with directed edges $(a,b)$ with the skew-symmetric condition $\sigma_{ab}=-\sigma_{ba}$.
\end{definition}

\begin{definition}[Triplet-orientation. Complete orientation.]
\label{def:triplet-complete-orientation} A graph triplet-orientation ${\bm\varsigma}$ is a set of numbers $\varsigma_{ac}^{(b)}=\pm 1$, associated with the triplets $a\to b\to c$ with $\varsigma_{ac}^{(b)}=-\varsigma_{ca}^{(b)}$. A set $({\bm\sigma},{\bm\varsigma})$, with ${\bm\sigma}$ being a graph edge orientation (see Definition~\ref{def:edge-orientation}) is called a complete orientation
\footnote{In the following we will use the term {\it orientation} to describe {\it complete orientation}.}.
\end{definition}
\begin{remark} Note that there are no triplets $a\to b\to c$, associated with the nodes $b\in{\cal G}_{0}$ of valence (degree) one. Therefore, the triplet-orientation components $\varsigma_{ac}^{(b)}$, associated with the nodes $b$ of degree one are represented by an empty set.
\end{remark}

{\bf G$^3$ model.} The partition function of the G$^3$ model, dependent on a complete orientation of the graph, $({\bm \sigma},{\bm \varsigma})$, and also dependent on the vector of weights ${\bm W}$, is defined in terms of the following integral (sum) over the Grassmann anti-commuting, variables $\varphi$:
\begin{eqnarray}
\label{VG3} Z_{\scriptsize\mbox{G}^3}({\bm \sigma},{\bm \varsigma};{\bm W})&=& \frac{\int
\exp\left(\frac{1}{2}\sum\limits_{(b\to a\to c)\in{\cal
G}_1}\varphi_{ab}\varsigma^{(a)}_{bc}W^{(a)}_{bc}\varphi_{ac}\right)
\exp\left(\frac{1}{2}\sum_{(a,b)\in{\cal G}_1}\varphi_{ab}\sigma_{ab}\varphi_{ba}\right)
\prod\limits_{(a,b)}d\varphi_{ab}}{\int \exp\left(\frac{1}{2}\sum_{(a,b)\in{\cal
G}_1}\varphi_{ab}\sigma_{ab}\varphi_{ba}\right) \prod\limits_{(a,b)}d\varphi_{ab}},
\end{eqnarray}
where  $\forall (a,b),(c,d)\in{\cal G}_1\quad
\varphi_{ab}\varphi_{cd}=-\varphi_{cd}\varphi_{ab}$. There are two independent Grassmann variables per edge (one per directed edge), and the Berezin/Grassman integrals in Eq.~(\ref{VG3}) are defined according to
standard Berezin rules, $\forall (a,b)\in{\cal G}_1:\quad \int d\varphi_{ab}=0,\quad
\int\varphi_{ab} d\varphi_{ab}=1$.
%%%%%%%%%%%%%%%%%%%%%%%%%%%%%%%%%%%%%%%%%%%%%%%%%%%%%%%%%%%%%%%%%%%%%%%%%%%%%%%%%%%%%%%%%%%%%%%%%%
%\begin{definition}[G$^3$ model.]
%\label{def:G3}
%\end{definition}
%%%%%%%%%%%%%%%%%%%%%%%%%%%%%%%%%%%%%%%%%%%%%%%%%%%%%%%%%%%%%%%%%%%%%%%%%%%%%%%%%%%%%%%%%%%%%%%%%%

\begin{remark}
(a) The partition function of the G$^3$ model can also be interpreted as a result of the following
Gaussian statistical average (expectation value) over the Grassmann variables
\begin{eqnarray}
Z_{\scriptsize\mbox{G}^3}({\bm \sigma},{\bm \varsigma};{\bm W})= \left\langle\exp\left(\frac{1}{2}\sum\limits_{(b\to a\to c)\in{\cal
G}_1}\varphi_{ab}\varsigma^{(a)}_{bc}W^{(a)}_{bc}\varphi_{ac}\right)\right\rangle.\nonumber
\end{eqnarray}

\noindent
(b) The denominator in Eq.~(\ref{VG3}), which is $\pm 1$ by construction, is introduced for convenience to enforce the normalization condition,
$Z_{\scriptsize\mbox{G}^3}({\bm\varsigma};{\bm \sigma};{\bm 0})=1$.

\noindent (c) Significance of the G$^3$ model, as of any other Gaussian Grassman model, is in the fact that
its partition/generating function is a Pfaffian. Indeed,  the numerator in Eq.~(\ref{VG3}) is the
Pfaffian of  the skew-symmetric matrix (of the size twice the number of edges in the graph) with the
following elements
\begin{eqnarray}
H_{ij}=\left\{\begin{array}{cc}
\varsigma^{(a)}_{bc}W^{(a)}_{bc}, & i=(a,b)\ \&\ j=(a,c),\mbox{ where } b\neq c\sim a,\\
\sigma_{ab},& i=(a,b),\ \& \ j=(b,a).
\end{array}\right.
\label{H-matrix}
\end{eqnarray}
Thus calculating the partition function of the G$^3$ model is easy, as requiring at most $O(|{\cal G}_1|^3)$ operations.

\noindent (d) When the graph ${\cal G}$ is a tree, the numerator in Eq.~(\ref{VG3}) does not depend on ${\bm W}$, and therefore the partition function $Z_{\scriptsize\mbox{G}^3}({\bm \sigma},{\bm \varsigma};{\bm W})=1$ is trivial. Nontrivial dependence of the partition function on ${\bm W}$ appears due to the loops in the graph, as will be discussed in details later in the manuscript.
\end{remark}

\subsection{Equivalence between the WBG model and the G$^3$ model}
\label{subsec:theorem}
The main subject of this Subsection (and also of the manuscript) will consist in establishing equivalence between the WBG model and the G$^3$ model for a given (arbitrary) set of edge-weights, ${\bm W}$.
However,  the partition function of the former model depends of ${\bm W}$ only, while the partition function of the later model depends on both the
edge-weight vector ${\bm W}$ and on a complete orientation $({\bm \sigma},{\bm \varsigma})$ of the graph. Therefore, the aforementioned equivalence between the models is established only for some special choice(s) of the complete orientation which is detailed in some number of forthcoming definitions preceding the main theorem/statement.

\begin{definition}{\bf (Closed, immersed and embedded walk. Immersed and embedded orbit. Oriented edges and triplets of the walk/orbit.)}
\label{def:walk} A {\it closed walk} $C$ of length $l(C)$ on a graph ${\cal G}$ (not necessarily planar) is an ordered set $C=\left(c_{0},\ldots,c_{n}\right)$ of nodes $c_{j}\in {\cal G}_{0}$,  such that  $c_j$ and $c_{j+1}$ are adjacent for all $j=0,\cdots,n-1$, and $c_{n}=c_{0}$.
A closed walk is called {\it immersed} if $c_{j-1}\ne c_{j+1}, \; \forall j=1,\ldots,n$, i.e., the walk does not
involve backtracking events (hereafter we use a natural agreement $c_{n+1}=c_{1}$). A closed walk is called {\it embedded} if $c_{j}\ne c_{k}, \; \forall
j\ne k$, i.e., the path does not have self-intersections. An equivalence class of closed walks with
respect to the cyclic permutations of the path nodes is called an orbit. An {\it
immersed orbit} and an {\it embedded orbit} is an equivalence class, represented by an immersed closed walk and an embedded closed walk, respectively. The ordered pairs $(c_{j-1},c_{j})$ and triplets $(c_{j-1},c_{j},c_{j+1})$ for $j=1,\ldots,n$ are called the oriented edges and triplets of $C$, respectively.
\end{definition}
\begin{remark}
\label{Rem-immersed} %We consider triplets to be defined for immersed closed walks (or orbits).
Obviously any embedded closed walk (or orbit) is immersed. The definition (\ref{def:walk})
assumes that orbits are oriented, however we will also be discussing undirected orbits as equivalence classes of orbits with respect to orientation change, sometimes
without additional clarification (when it does not cause a confusion). An undirected embedded orbit
is synonymous to a simple cycle. The terms {\it immersed} and {\it embedded} are borrowed from
topology and are used here based on the fact that immersed and embedded closed walks can be viewed
as discrete counterparts of the loops, immersed and embedded, respectively into a plane. An
immersed loop into a plane, or equivalently an immersion $S^{1}\looparrowright \mathbb{R}^{2}$ of a
circle into a plane is a smooth closed trajectory with everywhere non-zero velocity. This analogy
is detailed in Appendix \ref{sec:topology} and applied to unveil the topology that stands behind
the equivalence of the WBG and G$^3$ models.
\end{remark}

{\bf Sign of an immersed orbit, associated with a given complete orientation.} Given a complete graph orientation
$({\bm\sigma},{\bm\varsigma})$, we associate a sign $\varepsilon(C)=\varepsilon_{{\bm\sigma},{\bm\varsigma}}(C)=\pm 1$ with any immersed
orbit $C$ by
\begin{equation}
\label{define-epsilon} \varepsilon(C)=\left(\prod\limits_{(b\to a\to c)\in
C}\varsigma_{bc}^{(a)}\right)\prod\limits_{(a,b)\in C}\sigma_{ab}=
\prod_{j=0}^{l(C)-1}\varsigma_{c_{j-1}c_{j+1}}^{(c_{j})}\prod_{j=0}^{l(C)-1}\sigma_{c_{j}c_{j+1}}.
\end{equation}
%%%%%%%%%%%%%%%%%%%%%%%%%%%%%%%%%%%%%%%%%%%%%%%%%%%%%%%%%%%%%%%%%%%%%%%%%%%%%%%%%%%%%%%%%%%%%%%%%%
%\begin{definition}[Sign of an immersed orbit with given complete orientation.]
%\label{def:e(C)}
%\end{definition}
%%%%%%%%%%%%%%%%%%%%%%%%%%%%%%%%%%%%%%%%%%%%%%%%%%%%%%%%%%%%%%%%%%%%%%%%%%%%%%%%%%%%%%%%%%%%%%%%%%
\begin{remark}
\label{Rem-1} Note that the binary function $\varepsilon$ is obviously invariant with respect to
cyclic permutations of $C$, therefore, the signs $\varepsilon(C)$ are defined for oriented immersed
orbits. A change of orientation changes the sign of each factor in the products in the rhs of
Eq.~(\ref{define-epsilon}), and since the number of the $\sigma$-factors is the same as the number
of $\varsigma$-factors, $\varepsilon(C)$ is invariant with respect to the orientation changes.
Therefore, $\varepsilon$ is actually correctly defined for a non-oriented immersed orbit as well.
\end{remark}

{\bf Decomposition of a cycle.} A generic $\mathbb{Z}_{2}$-cycle can be decomposed into a set of immersed orbits in a variety of ways. It is easy to realize that each decomposition can be labeled by a set of variables ${\bm\xi}=\{\xi_{a}\}_{a\in {\cal G}_{0}}$, referred to as {\it partition data} where $\xi_{a}$ for $a\in \gamma$ denotes a partition of the set of the adjacent edges $\{a,b\}\in \gamma$ (or equivalently adjacent nodes $b\in \gamma$) into distinct non-ordered pairs. Formally $\xi_{a}$ can be viewed as a set whose elements are given by the aforementioned distinct pairs. The partition data ${\bm\xi}$ naturally provides a set $\{C_{1}(\gamma,{\bm\xi}),\ldots,C_{n(\gamma,{\bm\xi})}(\gamma,{\bm\xi})\}$ of immersed orbits that decompose $\gamma$. The partition variables $\xi_{a}$ have been already used earlier to express the higher vertex weights of a Wick binary model in terms of the lowest vertex weights, see Eqs.~(\ref{Z_WBG}) and
(\ref{W-expressions}).
\begin{definition}{\bf (Decomposed cycle. Decomposition Components. Total self-intersection index of a decomposed cycle.)}
\label{def:decomposed cycle} A pair $(\gamma,{\bm\xi})$, where $\gamma$ is a $\mathbb{Z}_{2}$-cycle and ${\bm\xi}$ is the partition data, associated with its vertices is called a decomposed cycle. The set $\{C_{1}(\gamma,{\bm\xi}),\ldots,C_{n(\gamma,{\bm\xi})}(\gamma,{\bm\xi})\}$ is called a decomposition of $\gamma$ associated with ${\bm\xi}$. The immersed orbits $C_{j}(\gamma,{\bm\xi})$ are referred to as the decomposition components. The total self-intersection index of a decomposed cycle $N(\gamma,{\bm\xi})\in \mathbb{Z}_{2}$ is defined as a sum $N(\gamma,{\bm\xi})=\sum_{a\in \gamma}N_{a}(\gamma,{\bm\xi})=\sum_{a\in \gamma}\sum_{\{\alpha,\alpha'\}}^{\alpha \ne \alpha'\in\xi_{a}}I_{\alpha}\cdot
I_{\alpha'}$ over the cycle vertices, where the (elementary) intersection index $I_{\alpha}\cdot I_{\alpha'}$  was defined in Subsection~\ref{subsec:WBG_def}.
% remind about the definition of I and I\cdot I, and the
\end{definition}
\begin{remark} (a) The definition of the local intersection index $N_{a}(\gamma,{\bm\xi})$ is fully consistent with Eq.~(\ref{intersect-index-node}), since $N_{a}(\gamma,{\bm\xi})=N_{a}(\gamma_{a},\xi_{a})$, where $\gamma_{a}$ is the set of neighbors of $a$ that belong to $\gamma$.

\noindent (b) Note that the decomposition components are immersed orbits of a special type: they do not have self-intersecting and intersecting edges, and all intersections and self-intersections (if any) occur at the vertices only. In particular, the components generate $\mathbb{Z}_{2}$-cycles that will be denoted in the same way, and the decomposition can be also represented as $\gamma=\sum_{j=1}^{n(\gamma,{\bm\xi})}C_{j}(\gamma,{\bm\xi})$.
\end{remark}

The just introduced notion of a decomposed cycle allows a convenient definition of a {\it suitable} orientation.
\begin{definition}[Suitable Complete Orientation.]
\label{def:suitable} A (complete) orientation is called suitable if for any decomposed cycle $(\gamma,{\bm\xi})$ the following relation holds
\begin{eqnarray}
\label{basic-property} (-1)^{q(\gamma,{\bm\xi})}=1,
\end{eqnarray}
where the binary function $q(\gamma,{\bm\xi})$ is defined by
\begin{eqnarray}
\label{define-q}
(-1)^{q(\gamma,{\bm\xi})}=(-1)^{N(\gamma,{\bm\xi})}\prod_{j=1}^{n(\gamma,{\bm\xi})}
\left(-\varepsilon(C_{j}(\gamma,{\bm\xi}))\right).
\end{eqnarray}
\end{definition}
Note that, as discussed later in Section \ref{sec:suitable}, such suitable complete orientations do exists.

We are now in the position to state the main theorem.
\begin{theorem}
\label{Theorem} For any connected planar graph (of an arbitrary degree) the discrete-variable Wick model
(\ref{Z_WBG}) and fermion model (\ref{VG3}) are term-by-term, i.e. fully,  equivalent to each
other, i.e., $Z_{\scriptsize\mbox{G}^3}({\bm\varsigma},{\bm \sigma};{\bm W})=Z_{\scriptsize
\mbox{WBG}}({\bm W})$, for any ${\bm W}$, provided $({\bm\sigma},{\bm\varsigma})$ is a suitable orientation of ${\cal G}$.
\end{theorem}

\section{Proof of the WBG and G$^3$ equivalence}
\label{sec:WBG=G3}

The main task of this Section is to sketch the proof of the Theorem \ref{Theorem}. We start by making the following (almost obvious) statement that reduces the proof of Theorem \ref{Theorem} to the case of irreducible graphs ${\cal G}$.
\begin{proposition}
\label{prop:gen->irreducible}
(i) The partition functions $Z_{\scriptsize \mbox{WBG}}({\bm W})$ and $Z_{\scriptsize\mbox{G}^3}({\bm\varsigma},{\bm \sigma};{\bm W})$ of the WBG and G$^3$ models both satisfy the factorization property: the partition function for ${\cal G}$ is given by the product of the partition functions of the corresponding models restricted to the irreducible components of ${\cal G}$. (ii) An orientation $({\bm\sigma},{\bm\varsigma})$ on ${\cal G}$ is suitable if and only if its restriction to any irreducible component of ${\cal G}$ is suitable.
\end{proposition}
\begin{proof} For the WBG model both the factorization property (i) and the feature (ii) follow immediately from the fact that any $\mathbb{Z}_{2}$-cycle on ${\cal G}$ is restricted to the disjoint union of its irreducible components. The factorization property for the G$^{3}$ model is obtained by consecutive integrations of the numerator and denominator of Eq.~(\ref{VG3}) over the pairs $d\varphi_{ab}d\varphi_{ba}$ of Grassmann variables related to the connecting edges $\{a,b\}$. Each integration is performed by representing, $\exp(\varphi_{ab}\varphi_{ba}/2)=1+\varphi_{ab}\varphi_{ba}/2$. Substitution of the first term into the numerator/denominator of Eq.~(\ref{VG3}) results in a product of two Grassmann integrals of even functions over an odd number of integrations, which yields zero. Substitution of the second term, followed by the integration over $d\varphi_{ab}d\varphi_{ba}$ results in the product of two integrals over two unconnected graphs obtained by eliminating the connecting edge $\{a,b\}$.
\end{proof}

Now we are in the position to describe a\\
{\bf Proof of Theorem~\ref{Theorem}:} Substituting Eq.~(\ref{W-expressions}), which expresses the higher order vertex functions
in terms of the respective edge terms, into Eq.~(\ref{Z_WBG}), we obtain the following expression for the partition function in the form of a sum over
the combined variables $(\gamma,{\bm\xi})$
\begin{eqnarray}
\label{Z-W-expand} Z_{\scriptsize \mbox{WBG}}({\bm
W})=\sum_{(\gamma,{\bm\xi})}(-1)^{N(\gamma,{\bm\xi})}
\prod_{j=1}^{n(\gamma,{\bm\xi})}r(C_{j}(\gamma,{\bm\xi})), \;\;\; r(C)=\prod_{k=1}^{l(C)}W_{c_{k-1}c_{k+1}}^{(c_{k})}.
\end{eqnarray}
Eq.~(\ref{Z-W-expand}) is derived by means of re-grouping the
factors correspondent to the lowest weights $W_{ac}^{(b)}$ and keeping together terms
correspond to the same immersed orbits $C_{j}(\gamma,{\bm\xi})$ participating in the
decomposition of $\gamma$ determined by  ${\bm\xi}$. The sign factor in Eq.~(\ref{Z-W-expand})
is given by the product of the
respective signs associated with the nodes of $\gamma$, and it is thus determined by the total
number of intersections, $N(\gamma,{\bm\xi})$, introduced in Definition~(\ref{def:decomposed cycle}).

A similar to Eq.~(\ref{Z-W-expand}) expansion for the partition function of the Fermion model is
\begin{eqnarray}
\label{Z-F-expand} Z_{\scriptsize\mbox{G}^3}({\bm\varsigma},{\bm \sigma};{\bm
W})=\sum_{(\gamma,{\bm\xi})}
\prod_{j=1}^{n(\gamma,{\bm\xi})}\left(-\varepsilon(C_{j}(\gamma,{\bm\xi}))\right)
\prod_{j=1}^{n(\gamma,{\bm\xi})}r(C_{j}(\gamma,{\bm\xi})).
\end{eqnarray}
It can be rationalized as follows. We represent the exponential under the average in Eq.~(\ref{VG3}) as a product of the vertex exponentials (labeled by $a$) and expand the vertex exponentials in the natural
bilinear combinations of the Grassmann variables $\varphi$. Each term of the expansion is naturally
labeled by a set ${\bm\xi}$ of partition variables. We further compute the expectation value of
each individual term in the expansion, using the Wick's theorem. Since the two-point correlation
functions of the Grassmann variable is,
$\left\langle\varphi_{ab}\varphi_{ba}\right\rangle=\sigma_{ab}$, when $\{a,b\}\in {\cal G}_{1}$ and,
$\left\langle\varphi_{ab}\varphi_{ba}\right\rangle=0$, otherwise, the set ${\bm\xi}$ of the partition
variables provides a non-zero contribution to the partition function if and only if it satisfies
the following property: If $\{a,b\}$ participates in the local partition at node $a$, then
$\{b,a\}=\{a,b\}$ participates at the local partition at node $b$. For a partition ${\bm\xi}$ that
satisfies this property we can build the associated cycle $\gamma$ consisting of edges
$\{a,b\}$ that participate in the local partitions. We further re-group the factors
$\varphi_{ab}\varsigma_{bc}^{(a)}W_{bc}^{(a)}\varphi_{ac}$ by keeping together the terms that
correspond to the same immersed orbits $C_{j}(\gamma,{\bm\xi})$ participating in the
decomposition of $\gamma$. The decomposition is
determined by  ${\bm\xi}$, followed by evaluating the expectation values
using the Wick's theorem, in particular with making use of the form of the two-point correlation
functions. After re-grouping the factors correspondent to each immersed orbit, and combining together the
$W_{ac}^{(b)}$, $\sigma_{ab}$ and $\varsigma_{ac}^{(b)}$ factors, this results
in Eq.~(\ref{Z-F-expand}).

By Definition~\ref{def:suitable} the sign factors in all the individual contributions to the partition functions
$Z_{\scriptsize \mbox{WBG}}({\bm W})$ and
$Z_{\scriptsize\mbox{G}^3}({\bm\varsigma},{\bm \sigma};{\bm W})$ (labeled by respective suitable orientation,
$(\gamma,{\bm\xi})$) are the same, which proves the statement of the theorem.
\qed

\section{Construction of a suitable orientation}
\label{sec:suitable}

In this Section we show how to construct efficiently a {\it suitable orientation}, required to compute the partition function of the WGB model efficiently, according to the Theorem \ref{Theorem}.  The construction of the suitable orientation is split into two steps.  First step, discussed in Section \ref{subsec:extended-induced}, describes transformation of the original irreducible graph to the respective extended graph, ${\cal G}^e$, which is a homogeneous graph of degree three with even number of vertices. We also show in this Subsection how a Kasteleyn orientation of edges on the extended graph generates a so-called {\it induced orientation} on the original graph. Second step, undertaken in Section \ref{subsec:induced-suitable}, focuses on showing that an induced orientation of the original graph is always suitable, i.e. satisfies the condition of Theorem~\ref{Theorem}.

Note that Proposition~\ref{prop:gen->irreducible} reduces the problem of identifying suitable orientations to the case when ${\cal G}$ is irreducible. Any irreducible graph of degree two is a cyclic graph, for which suitable orientations are identified in an obvious way. Therefore, throughout this section the graph ${\cal G}$, on which our models are defined is assumed irreducible and of degree three or higher (even when not explicitly stated), which does not lead to any loss of generality.

\subsection{Extended Graph and Induced Orientation}
\label{subsec:extended-induced}

\begin{figure}
\includegraphics[width=0.5\textwidth,page=1]{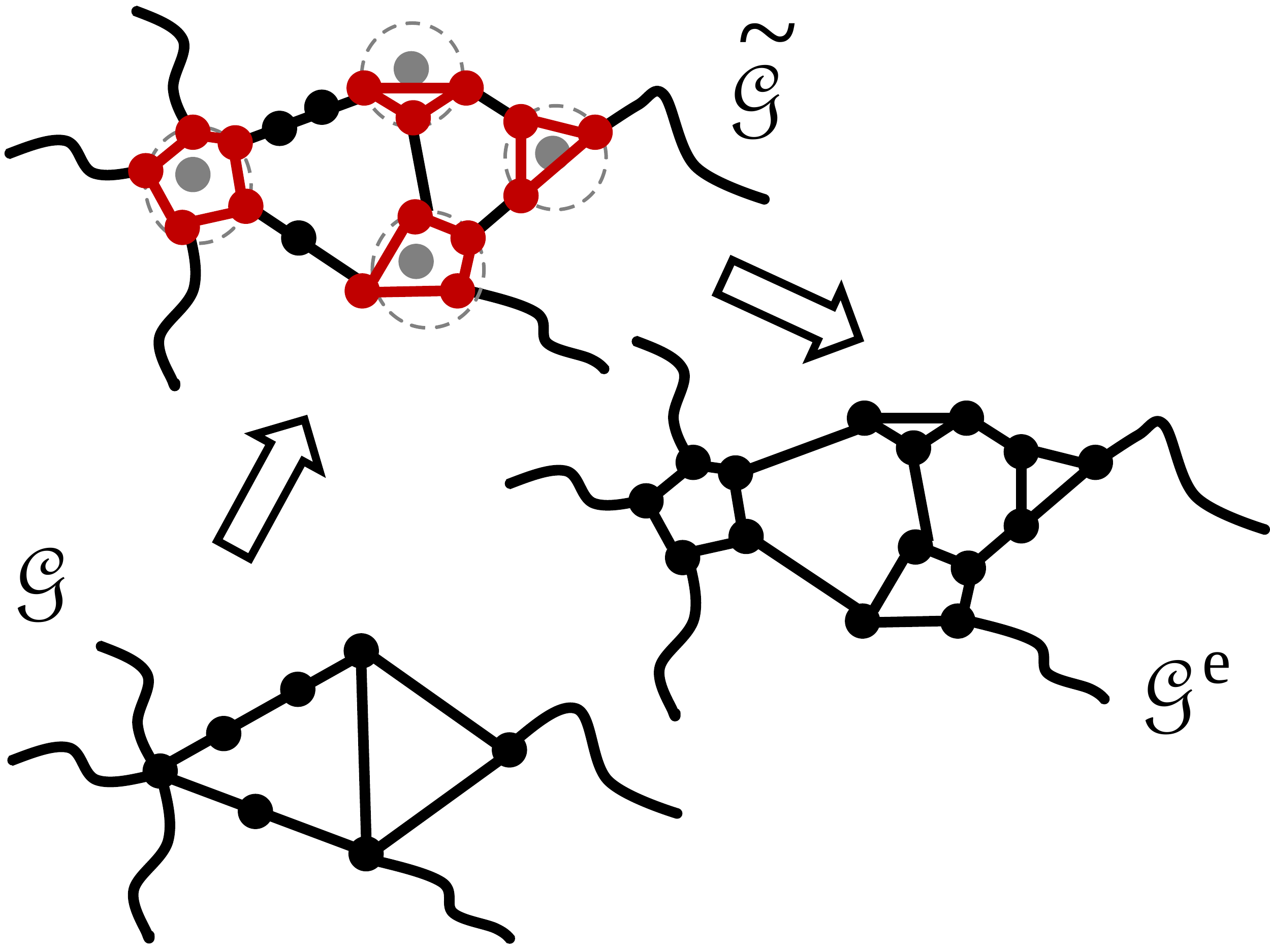}
\caption{ Construction of the extended graph ${\cal G}^e$ associated with ${\cal G}$. A degree three graph $\tilde{\cal G}$ is constructed on an intermediate step. $\tilde{\cal G}$ is obtained from ${\cal G}$ by extending vertices. ${\cal G}^e$ is obtained from ${\cal G}$ by replacing linear segments with elementary edges. \label{fig:New1}}
\end{figure}

Construction of the extended graph, ${\cal G}^e$, illustrated in Fig.~(\ref{fig:New1}), is done in two steps. We start from an original irreducible graph ${\cal G}$ of degree three or higher, which, in particular does not have vertices of degree one, and extend its vertices of degree three or higher to  circles (or polygons) so that the adjacent edges of the original graph are connected to the polygon vertices. The procedure results in an irreducible graph $\tilde{{\cal G}}$ of degree three,  which is however not necessarily homogeneous degree three, as it may have linear segments consisting of vertices of degree two. We further replace the linear segments of $\tilde{{\cal G}}$ with elementary edges to obtain an irreducible homogeneous graph ${\cal G}^e$ of degree three.

\begin{definition}[Extended Graph. Extending Polygons.]
\label{def:extended} The constructed above (in the preceding paragraph) irreducible homogeneous graph ${\cal G}^e$ of degree three
is called the extended graph associated with an irreducible graph ${\cal G}$. The polygons that replace the vertices of degree three and higher are called the extending polygons.
\end{definition}

We continue working the irreducible degree three homogeneous degree three graph  ${\cal G}^e$ and observe that by construction it has an even number of vertices. Thus, following Kasteleyn \cite{63Kas,67Kas} (see also \cite{86LP} for a comprehensive review) we construct the so-called Kasteleyn edge orientation on the graph. (See also Fig.~\ref{fig:edge-triplet}a for an illustration.)
\begin{definition}[Kasteleyn edge orientation.]
\label{def:Kasteleyn-edge}
An edge orientation of ${\cal G}^e$ (which is a planar graph with even number of vertices) is called Kasteleyn edge orientation of the graph if for any (even or odd) face of the graph (possibly under exception of the outer face) the number of clockwise-oriented (negative) edges is odd.
This is formalized in terms of the vector, ${\bm \sigma}^e=(\sigma_{ab}^e=\pm 1|\{a,b\}\in{\cal G}^e)$, where $\sigma_{ab}^e=+1$ if the arrow/orientation assigned to the undirected edge $\{a,b\}$ goes from $a$ to $b$,  and $\sigma_{ab}^e=-1$ otherwise.
\end{definition}

\begin{figure}
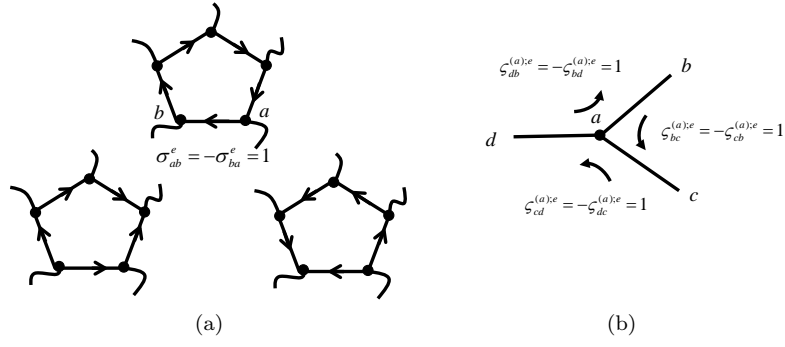
\centering
\vspace{-0.5cm}\subfigure[]{\includegraphics[width=0.3\textwidth,page=2]{NewPlanarFigs.pdf}}
\subfigure[]{\includegraphics[width=0.3\textwidth,page=3]{NewPlanarFigs.pdf}}
\caption{Fig.~(a) illustrates construction of the Kasteleyn orientation of edges - the three examples (with one, three and five clockwise edges per
face respectively)  are all legitimate Kasteleyn orientations of the selected internal face. Fig.~(b) illustrates construction of a triplet orientation of a vertex of degree three (i.e. a vertex of ${\cal G}^e$).\label{fig:edge-triplet}}
\end{figure}

Furthermore, ${\cal G}^e$ is also equipped with a natural triplet orientation, built using the following rule: we set $\varsigma_{bc}^{(a);e}=1$, if going from $b$ to $c$ represents a left most turn at $a$, and $\varsigma_{bc}^{(a);e}=-1$ otherwise. (See Fig.~\ref{fig:edge-triplet}b for an illustration.) ${\bm \varsigma}$ is a vector constructed from all the triplet orientations of ${\cal G}^e$.
The apparently sloppy construction is actually rigorous, since a planar graph, in particular ${\cal G}^e$, is supplied with a cyclic (e.g., counterclockwise) ordering of edges, adjacent to any node.

\begin{definition}[Kasteleyn (Complete) Orientation.]
\label{def:extended} A constructed above triplet orientation ${\bm\varsigma}^{e}$ is called the left triplet orientation of ${\cal G}^e$. A (complete) orientation $({\bm\sigma}^{e},{\bm\varsigma}^{e})$ of an irreducible homogeneous graph ${\cal G}^e$ of degree three, represented by a Kasteleyn edge orientation ${\bm\sigma}^{e}$ and the left triplet orientation ${\bm\varsigma}$ is called a Kasteleyn (complete) orientation of ${\cal G}^e$.
\end{definition}

\begin{figure}
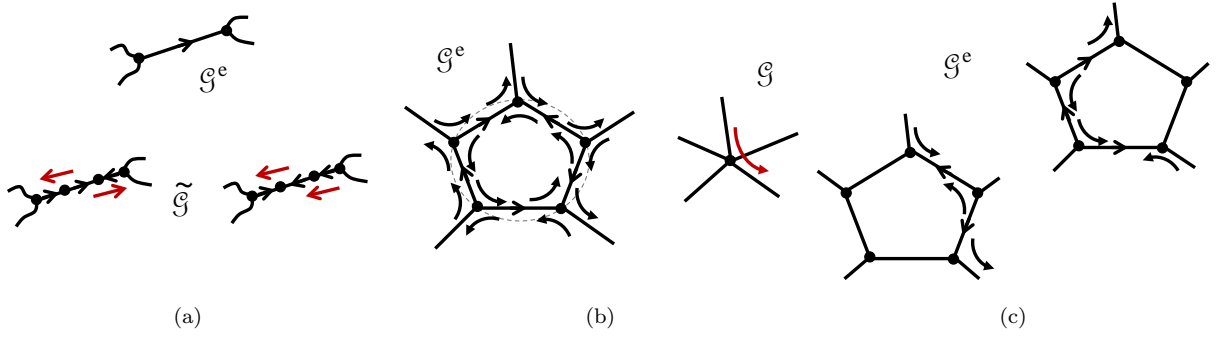
\centering
\subfigure[]{\includegraphics[width=0.3\textwidth,page=4]{NewPlanarFigs.pdf}}
\subfigure[]{\includegraphics[width=0.3\textwidth,page=5]{NewPlanarFigs.pdf}}
\subfigure[]{\includegraphics[width=0.3\textwidth,page=6]{NewPlanarFigs.pdf}}
\caption{Fig. (a) shows two examples of possible choice of edge and triplet orientations within a linear segment of $\tilde{\cal G}$.
The two examples are compatible, i.e. they both results according to Eq.~(\ref{linear-rule}) in the same orientation of the respective single edge in ${\cal G}^e$ replacing the linear segment in $\tilde{\cal G}$. The number
of arrows,  marking both edge and triplet orientations within the linear segment of $\tilde{\cal G}$, directed against the resulting arrow/orientation of the respective edge of ${\cal G}^e$ is even.
Fig. (b) shows fragment of complete Kasteleyn orientation on ${\cal G}^e$ and the related induced orientation of the respective triplet on ${\cal G}$. Fig. (c) illustrates how two equivalent ways of counting orientations on ${\cal G}^e$ results in the same triplet orientation
on ${\cal G}$,  correspondent to the triplet from Fig. (b): the total number of the upper/left orientation within either of the paths (shown on left and right figure) which are against the resulting (red) arrow on ${\cal G}$ is even.
 \label{fig:e-to-tilde}}
\end{figure}

{\bf Inducing orientations.} Our next task becomes to generate from $({\bm\sigma}^{e},{\bm\varsigma}^{e})$, defined on ${\cal G}^e$, a set of useful orientations $({\bm\sigma},{\bm\varsigma})$ on the original graph ${\cal G}$. This task is solved in a number of steps. We first build a set of orientations $(\tilde{{\bm\sigma}},\tilde{{\bm\varsigma}})$ on $\tilde{{\cal G}}$ by staying with the left triplet orientation for vertices of degree three and choosing the edge and triplet orientations on the linear segments of $\tilde{{\cal G}}$ (end vertices excluded) in an arbitrary way with the only constraint per linear segment that the product,
\begin{equation}
\tilde{\sigma}_{a_{0}a_{1}}\ldots \tilde{\sigma}_{a_{l-1}a_{l}}\tilde{\varsigma}_{a_{0}a_{2}}^{(a_{1})}\ldots \tilde{\varsigma}_{a_{l-2}a_{l}}^{(a_{l-1})}=\sigma_{a_{0}a_{l}}^{e},
\label{linear-rule}
\end{equation}
of all the edge an triplet sign factors reproduce the sign factor for the corresponding edge of ${\cal G}^e$. See Fig.~\ref{fig:e-to-tilde}a for illustration of the construction. Then,  we turn to constructing orientation on ${\cal G}$, first building respective components of $({\bm\sigma},{\bm\varsigma})$ for degree two vertices to be identical to the respective (and just constructed) orientations of the degree two vertices on $\tilde{\cal G}$. The components of ${\bm\varsigma}$ related to degree three and higher vertices of ${\cal G}$ are defined as the products of the edge and inner triplet orientation factors of ${\bm\sigma}^{e}$ ${\bm\varsigma}^{e}$, respectively, over the paths that connect the corresponding edges going over the extending polygons in either (counterclockwise or clockwise) direction, as shown in Fig.~\ref{fig:e-to-tilde}b,c. The Kasteleyn nature of the complete orientation $({\bm\sigma}^{e},{\bm\varsigma}^{e})$ ensures the independence of ${\bm\varsigma}$ on the choice of the directions of the connecting paths.
\begin{definition}[Induced Orientation.]
\label{def:induced-orientation} A (complete) orientation $({\bm\sigma},{\bm\varsigma})$ of an irreducible graph, constructed above (in the preceding paragraph) is called an induced (complete) orientation, or more specifically an orientation induced from a Kasteleyn edge orientation ${\bm\sigma}^e$ of ${\cal G}^e$.
\end{definition}

\begin{figure}\centering
\includegraphics[width=0.3\textwidth,page=7]{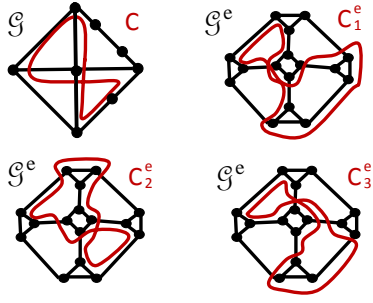}
\caption{Illustration of the lifting procedure.  Top left: an example of ${\cal G}$ (black) and an immersed orbit $C$ (red). Three other sub-figures
 show three different (but equivalent in the sense of their respective $\varepsilon_{{\bm\sigma}^e,{\bm\varsigma}^e}(C^e)$ contributions, see
 Proposition \ref{prop-lift-epsilon})
lifted version of $C$ on the respective extended graph, ${\cal G}^e$.
\label{fig:lift}}
\end{figure}

{\bf Lifting immersed orbits to the extended graph.} Induced orientations are constructed in a way that the $\varepsilon(C)$-factors for the immersed orbits in ${\cal G}$, defined by Eq.~(\ref{define-epsilon}), can be studied by {\it lifting} them to the extended graph.  Given an immersed orbit in ${\cal G}$ we first lift it to $\tilde{{\cal G}}$ by connecting the ends of the linear segments of $C$ with paths going over the extending polygons. By further replacing the linear segments with elementary edges we obtain a lifted orbit in ${\cal G}^e$. Obviously, by construction a lifted orbit, $C^e$, is not unique and is specified by the choice of the directions (counterclockwise/clockwise) of the connecting paths. Lifting procedure is illustrated in Fig.~\ref{fig:lift}.

\begin{proposition}
\label{prop-lift-epsilon}  Let $({\bm\sigma},{\bm\varsigma})$ on ${\cal G}$ be induced from a Kasteleyn orientation $({\bm\sigma}^e,{\bm\varsigma}^e)$ on ${\cal G}^e$, and $C^e$ be any lift of an immersed orbit $C$ of ${\cal G}$ to the extended graph ${\cal G}^e$. Then $\varepsilon_{{\bm\sigma}^{e},{\bm\varsigma}^{e}}(C^e)
=\varepsilon_{{\bm\sigma},{\bm\varsigma}}(C)$.
\end{proposition}
\begin{proof} The statement follows directly from the construction of the induced orientation. For $\varepsilon_{{\bm\sigma},{\bm\varsigma}}(C)$ the products of edge an inner triplet factors over the linear segments reproduce the edge factors ${\bm\sigma}^{e}$ of the corresponding edges of ${\cal G}^e$. The triplet factors of ${\bm\varsigma}$ for vertices of degree three or higher reproduce the products of the edge and inner triplet factors of $({\bm\sigma}^{e},{\bm\varsigma}^{e})$ over the connecting paths. The triplets that connect the linear segments to the connecting paths always come in pairs with two right/left turns for the counterclockwise/clockwise choice of the connecting path. \qed
\end{proof}

The important role of the Kasteleyn complete orientations and the corresponding induced complete orientations is described by the following statements.
\begin{proposition}
\label{P4} A Kasteleyn (complete) orientation of an irreducible homogeneous planar graph of degree three, ${\cal G}^e$
is suitable. According to Definition \ref{def:suitable}, this is obviously equivalent to the statement that $\varepsilon(C^e)=-1$ for any embedded curve (simple cycle) $C^e$ \footnote{We may consider any irreducible homogeneous graph of degree three here, and not necessarily the graph ${\cal G}^e$ extended from an arbitrary irreducible planar graph.  However,  since the proposition will be applied in the following solely to the extended graph, ${\cal G}^e$, we choose not to introduce new notations for an arbitrary  irreducible homogeneous graph of degree three, using ${\cal G}^e$ for this purposes here.}.
\end{proposition}

%-------------------------------------------------------------------------------------------------
\begin{figure} \includegraphics[width=0.4\textwidth,page=15]{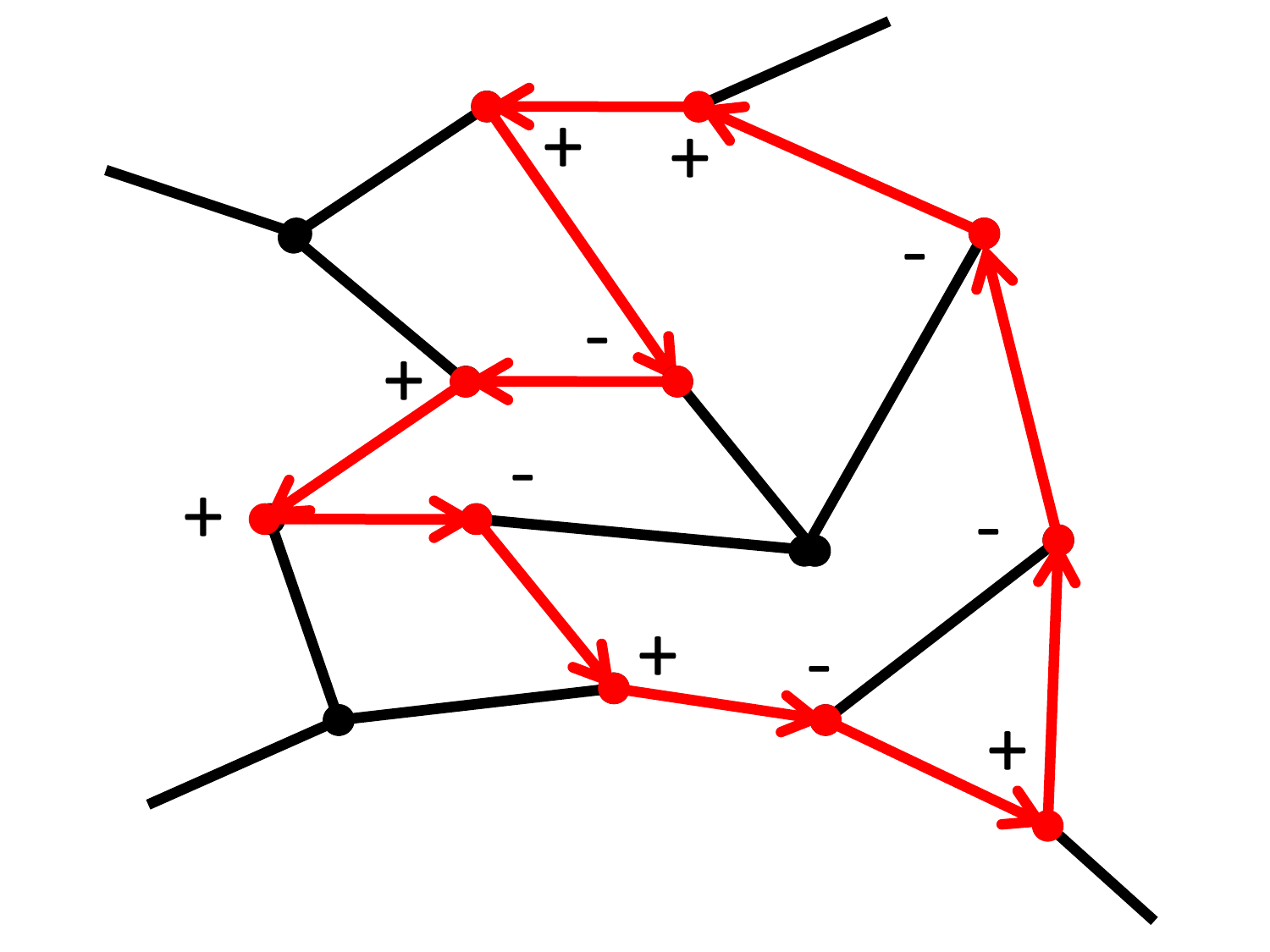} \caption{Illustration of the $\varepsilon_{{\bm\varsigma}^e}$ computation for an imbedded orbit of a homogeneous degree three graph, used in the proof of the Proposition \ref{P4}. Arrow along the marked (red) embedded orbit indicates anti-clockwise (positive) direction of the walk. $\pm $ , standing next to vertices along the embedded orbit (simple cycle), correspond to the values of respective $\varsigma$ contributions. The number of edges adjusted to the $\pm$ vertex and belonging to the interior of the cycle is two/three respectively. \label{fig:valence}} \end{figure}
%-------------------------------------------------------------------------------------------------

\begin{proof}
The simple cycle $C^e$ partitions the plane into two connected regions, interior and exterior respectively. Therefore,  $C^e=\partial{\cal M}$ is the boundary of the ``inside''
connected component (the one whose closure is compact), represented by a planar graph ${\cal
M}\subset{\cal G}^e$ (the boundary is included). It is enough to prove the statement for the counterclockwise
orientation of $C^e$, throughout this proof all involved orientations are counterclockwise. We further represent $\varepsilon(C^e)=\varepsilon_{{\bm\sigma}^e}\varepsilon_{{\bm\varsigma}^e}$ as a product of edge and triplet contributions along the embedded curve $C^e$, and denote by $n_{0}$, $n_{1}$, $n_{2}$ and $l$ the number of vertices, edges, and faces of ${\cal M}$ and the length (number of edges) of $C^e$, respectively. Considering the product of $\varepsilon_{{\bm\sigma}^e}(\partial M)=-1$ (the equality is due to the Kasteleyn nature of ${\bm\sigma}^e$) over all faces of
${\cal M}$  we arrive at $(-1)^{n_{2}}=(-1)^{n_{1}-l}\varepsilon_{{\bm\sigma}^e}(C^e)$, since each internal edge participates in the product exactly twice and with opposite orientations. Therefore, $\varepsilon_{{\bm\sigma}^e}(C^e)=(-1)^{n_{1}-n_{2}-l}$.
The triplet contribution is estimated $\varepsilon_{{\bm\varsigma}^e}(C^e)=(-1)^{n_{R}}$ with $n_{R}$ being the number of right turns of $C^e$. Since each internal node of $\cal M$ is of degree three, whereas a boundary node $a\in C_{0}$ has degree two/three (in ${\cal M}$) when
$C$ turns left/right at $a$, and $\varsigma^{(a);e}=\pm 1$, respectively (see Fig.~\ref{fig:valence} for an illustration). In view of the above argument we have for the number of edges $2n_{1}=3(n_{0}-l)+3n_{R}+2(l-n_{R})$. This results in $\varepsilon_{{\bm\varsigma}^e}(C^e)=(-1)^{n_{0}-l}$.
Combining the obtained results for the edge and triplet contribution we arrive at $\varepsilon(C^e)=(-1)^{n_{0}-n_{1}+n_{2}}=-1$, with the second equality due to the Euler characteristic relation, $n_{0}-n_{1}+n_{2}=1$. \qed
\end{proof}
\begin{corollary}
\label{cor:almost-embedded} If an immersed orbit $C$ in ${\cal G}$ is lifted to an embedded orbit in ${\cal G}^e$, then $\varepsilon_{{\bm\sigma},{\bm\varsigma}}(C)=-1$, provided $({\bm\sigma},{\bm\varsigma})$ is induced.
\end{corollary}
\begin{proof} Let $C^e$ be any lift that satisfies the condition of the corollary. Then $\varepsilon_{{\bm\sigma},{\bm\varsigma}}(C)=\varepsilon_{{\bm\sigma}^{e},{\bm\varsigma}^{e}}(C^e)=-1$ by Propositions~\ref{prop-lift-epsilon} and \ref{P4}.
\end{proof}

\subsection{Induced Orientations Are Suitable}
\label{subsec:induced-suitable}

The following statement explicitly provides a broad class of suitable orientations on an irreducible graph.
\begin{lemma}
\label{Lemma} Any induced (complete) orientation on an irreducible graph is suitable.
\end{lemma}
\begin{proof}
We are not presenting a complete detailed proof of Lemma~\ref{Lemma}, since it contains a large number of straightforward verification. We rather provide a sequence of steps that leads to a complete proof with brief explanations of these verifications involved on each step. The main idea is to demonstrate that for an induced orientation the binary function $q(\gamma,{\bm\xi})$, defined in the Definition \ref{def:suitable}, does not depend on the partition data ${\bm\xi}$ (Steps 1 and 2) and then explicitly compute it for some special choices of ${\bm\xi}$ (Step 3).

\begin{figure}
\centering
\includegraphics[width=0.4\textwidth,page=9]{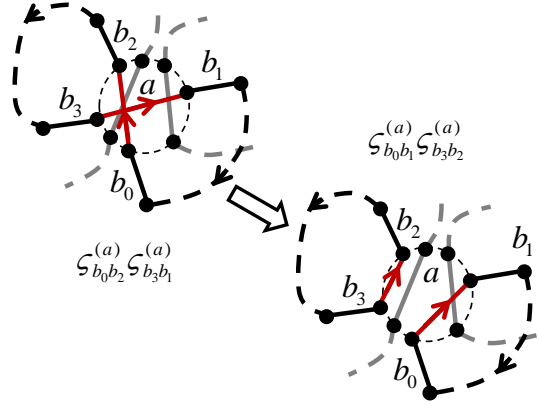}
\caption{ Example of the elementary transformation which eliminates an intersection of the segments $I_{\{b_0,b_2\}}$ and $I_{\{b_1,b_3\}}$ shown in red. The products of the triplet orientation factors involved remain the same according to Eq.~(\ref{triplet-suitable}),  while the total intersection number of the involved segments with any other (e.g. these shown in gray) remains unchanged modulo $2$. The change in the number of components,
$n(\gamma,{\bm\xi})$, by one is compensated by the identical (modulo $2$) change in the intersection index, $N_a(\gamma,{\bm\xi})$.\label{fig:varsigma_transf}}
\end{figure}

{\it Step 1.} We start with noting that the triplet component ${\bm\varsigma}$ of any induced orientation satisfies the following set of local constraints. Let $a$ be an arbitrary vertex of degree $4$ or higher, and $(b_{0},b_{1},b_{2},b_{3})$ be a set of its four arbitrary distinct neighbors (i.e., $b_{j}\sim a$), cyclically ordered in the counterclockwise direction with an arbitrary choice of the first neighbor. Then,
the following relation, between the pair-wise components of the left triplet orientations associated with $a$ and $(b_{0},b_{1},b_{2},b_{3})$ (see Definition \ref{def:extended}), holds
\begin{eqnarray}
\label{triplet-suitable} \varsigma_{b_{0}b_{2}}^{(a)}\varsigma_{b_{3}b_{1}}^{(a)}
=\varsigma_{b_{3}b_{2}}^{(a)}\varsigma_{b_{0}b_{1}}^{(a)}.
\end{eqnarray}
This can be verified directly. See Fig.~\ref{fig:varsigma_transf} for an illustration.

{\it Step 2.} We further introduce a set of local transformations of the decomposition data ${\bm\xi}$, illustrated in Fig.~\ref{fig:varsigma_transf}, and show that $q(\gamma;{\bm\xi})$ stays invariant under these transformations.
Modulo $2$ the change in $N(\gamma,{\bm\xi})$ is totally determined by the change of the intersection index of the edges involved in an elementary transformation. The change in the product of the edge and triplet factors involved in the product of all $\varepsilon$-factors in Eq.~(\ref{define-q}) boils down to the change in the product of two triplet factors involved in the elementary transformation. The change in the number $n(\gamma,{\bm\xi})$ of the decomposition components can be also analyzed locally, as illustrated in Fig.~\ref{fig:varsigma_transf}. Therefore, the invariance of $(-1)^{q(\gamma,{\bm\xi})}$ under an elementary transformation can be verified locally and boils down to the validity of Eq.~(\ref{triplet-suitable}).

Finally, we note that by performing the described elementary transformation each decomposition data ${\bm\xi}$ can be transformed to a preferred one, shown in Fig.~\ref{fig:pref}, where for each $a\in \gamma$ the adjacent edges are paired with their nearest neighbors (counterclockwise or clockwise) on $\gamma$. Obviously, any preferred decomposition satisfies the property $N_{a}(\gamma,{\bm\xi})=0$ for any $a\in \gamma$.

\begin{figure}
\centering
\includegraphics[width=0.4\textwidth,page=8]{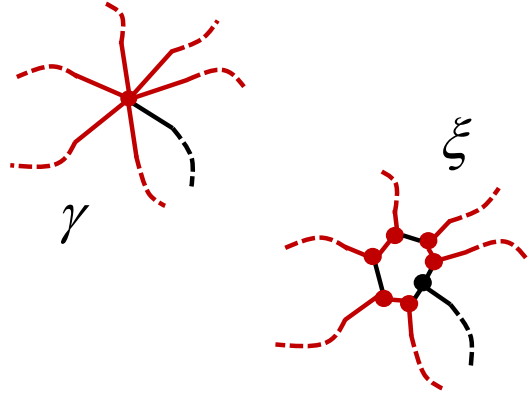}
\caption{ A single vertex element of an even generalized loop ($\mathbb{Z}_2$-cycle) $\gamma$ of ${\cal G}$ is shown in red at the top/left subfigure. Respected elements of the preferred decomposition data $\xi$ (containing no intersections) lifted to ${\cal G}^e$ is shown in red at the bottom/right subfigure.
\label{fig:pref}}
\end{figure}

{\it Step 3.} The previous step has reduced the proof to the case when the partition data of a decomposed cycle $(\gamma,{\bm\xi})$ is of a preferred form. We further lift the decomposition components $C_{j}(\gamma,{\bm\xi})$ to ${\cal G}^e$ by choosing the connection paths on the extending polygons in a way that they do not intersect with each other as shown in Fig.~\ref{fig:pref}, which is always possible for our special choice of ${\bm\xi}$. Therefore the immersed orbits $C_{j}(\gamma,{\bm\xi})$ satisfy the conditions of Corollary~\ref{cor:almost-embedded}, which results in $q(\gamma,{\bm\xi})=(-1)^{N(\gamma,{\bm\xi})}$. Also, as noted above, we have $N_{a}(\gamma,{\bm\xi})=0$ for all nodes $a\in\gamma$, so that $N(\gamma,{\bm\xi})=0$. \qed
\end{proof}

We would like to emphasize that Lemma~\ref{Lemma} not only proves the existence of suitable orientations, but rather provides a practical way of building them by reducing the problem to finding Kasteleyn edge orientations on the extended graph.

\section{Gauge Transformations}
\label{sec:gauge}

Gauge transformation approach discussed in \cite{06CCa,06CCb,07CC},
keeping the partition function invariant and the graph intact but modifying the factor functions,
also applies to the WBG model. On the other hand, and as shown above, the WBG model is det easy, and thus any binary model
derived from WBG models via application of a gauge transformation is also easy. This explains the logic of the main part of this Section --- to describe the extended family of binary models reducible via gauge transformations to WBG models. Section \ref{subsec:gauge-reminder} serves for a brief review of the gauge transformations. Sections \ref{subsec:dimer-example},\ref{subsec:ice},\ref{subsec:Ising} describe the gauge transformations reducing the three noticeable examples of dimer, ice and Ising models to respective WBG models.

Section \ref{subsec:X-matching} describes a complementary example of an $\#X$-matching det-tractable model, introduced in \cite{02Val,04Val,08Val}, which cannot be reduced (to the best of our knowledge) via a standard binary Gauge transformation to a WBG model, but it is still reducible to a WBG model on a contracted graph. The reduction to an easy model is shown here with the help of generalization to  intermediate model defined over $3$-ary (non-binary) alphabet.

\subsection{Gauge Transformations for general binary model. Reminder.}
\label{subsec:gauge-reminder}

Let us first define a general Edge-Binary Graphical (EBG) model on ${\cal G}$ with the
following partition function
\begin{eqnarray}
Z_{EBG}=\sum_{\bm\pi}\prod_{a\in{\cal G}_0}f_a({\bm\pi}_a),
\label{EBG}
\end{eqnarray}
where ${\bm\pi}=(\pi_{ab}=\pi_{ba}=0,1|\{a,b\}\in{\cal G}_1)$, and $\forall a\in{\cal G}_0$,
${\bm\pi}_a=(\pi_{ab}|b\in{\cal G}_0;\{a,b\}\in{\cal G}_1)$ and $f_a({\bm\pi}_a)$ are the cost
functions associated with vertices.

As shown and discussed in \cite{06CCa,06CCb,07CC} the partition function of the general EBG model (\ref{EBG})
is invariant under the set of linear, so-called Gauge transformations,
\begin{eqnarray}
f_a({\bm\pi}_a)\to \tilde{f}_a({\bm\pi}_a)=\sum_{{\bm\pi}_a'}\left(\prod\limits_{b\in
a}G_{ab}(\pi_{ab},\pi_{ab}')\right) f_a({\bm\pi}'_a), \label{Gauge1}
\end{eqnarray}
where the newly introduced (two per edge of the graph) Gauge matrices satisfy the following skew
orthogonality conditions
\begin{eqnarray}
\forall \{a,b\}\in{\cal G}_1:\quad \sum_\pi G_{ab}(\pi,\pi')G_{ba}(\pi,\pi'')=\delta(\pi',\pi'').
\label{skew}
\end{eqnarray}
With the gauge transformation applied, the partition function of the model becomes
\begin{eqnarray}
Z_{EBG}=\sum_{\bm\pi}\prod_{a\in{\cal G}_0}\tilde{f}_a({\bm\pi}_a)= \sum_{\bm\pi}\prod_{a\in{\cal
G}_0}\left(\sum_{{\bm\pi}_a'}\left(\prod\limits_{b\sim a}G_{ab}(\pi_{ab},\pi_{ab}')\right)
f_a({\bm\pi}_a')\right). \label{Z_EBG}
\end{eqnarray}
The expression in the middle of Eq.~(\ref{Z_EBG}) may be considered as a new graphical model.
Notice that if all gauge-matrices are non-singular (which is the case discussed in this manuscript)
transformation from (\ref{EBG}) to (\ref{Z_EBG}) is invertible,  i.e. for any set of
$G$-transformation one naturally defines the reverse $G^{-1}$, bringing us back to the original
formulation of Eq.~(\ref{Z_EBG}).

In \cite{06CCa,06CCb,07CC,08CCT} the Gauge transformation was discussed in the context of the
so-called Belief Propagation (BP)-gauge and the resulting Loop Calculus/Series,  where BP
constraints were imposed additionally to the general (and always maintained) skew-orthogonality
constraints (\ref{skew}). The main use of the gauge transformations in the preceding publications
consisted in fixing the gauge freedom even further according to the so-called Belief Propagation
(BP) conditions. The BP equations,  enforcing the respective gauge fixing conditions correspond to
fixed points of the Belief Propagation message-passing algorithm popular in physics, computer
science and information theory. This special, BP gauge, reduces the number of terms on the rhs of
Eq.~(\ref{Z_EBG}) requiring that only terms correspondent to $\pi=1$-edges forming a generalized
loop (i.e. subgraph with all vertices of degree at least two). The resulting series for partition
function consisting of generalized loop only, is called Loop Series.

In the following we show how some set of binary models on the planar graph,  which are known to be
easy from previous studies, are reduced to WBG forms under respective gauge transformation.

\subsection{Example: Dimer model as a WBG model}
\label{subsec:dimer-example}

Gauge transformation allows to restate the original model in a new basis, parameterized by $G$,
which is generally a useful freedom to use. To illustrate its general utility, we discuss how
applying a gauge transformation allows to restate the dimer model as a WBG model on
the same graph. First, let us restate DM in the general EBG form
\begin{eqnarray}
f^{(dm)}_a({\bm\pi}_a)=\prod_{b\sim a} w_{ab}^{\pi_{ab}/2}*\left\{\begin{array}{cc}
1,& \sum_{c\sim a}\pi_{ac}=1,\\ 0, & \mbox{otherwise}.
\end{array}\right.
\label{f_DM}
\end{eqnarray}
We will build here a gauge transformation around a valid dimer (perfect matching) configuration,
${\bm p}=(p_{ab}=0,1| \forall a:\quad \sum_{b\sim a} p_{ab}=1)$. Then introduce the following ${\bm p}$-dependent gauge ($2\times 2$ matrix with
elements $G_{ab}(\pi_{ab},\pi_{ab}')$ parametrically dependent on ${\bm p}$)
\begin{eqnarray}
G_{ab}=\left(\begin{array}{cc} 0 & 1\\ 1 & 0\end{array}\right)\quad \mbox{if}\quad
p_{ab}=p_{ba}=1,\quad \mbox{or else}\quad G_{ab}=\left(\begin{array}{cc} 1 & 0\\ 0 & 1\end{array}\right).
\label{dim-Gauge}
\end{eqnarray}
Obviously this gauge satisfies the skew-orthogonality condition (\ref{skew}). Furthermore, one finds that
\begin{eqnarray}
\tilde{f}^{(dm)}_a({\bm\pi}_a;{\bm p}_a)=\left\{
\begin{array}{cc}
\sqrt{w_{ab}}, & \sum_{c\sim a}\pi_{ac}=0\\
\sqrt{w_{ac}}, & \pi_{ab}=\pi_{ac}=p_{ab}=1,\quad\&\quad \sum_{d\sim a}^{d\neq b,c}\pi_{ad}=0,\\
0, & \mbox{otherwise},
\end{array}\right.
\label{tf_DM}
\end{eqnarray}
where $\{a,b\}$ is the special (unique for vertex $a$) matching edge, i.e. $p_{ab}=1$. Therefore DM model (\ref{f_DM}) is reduced to the WBG model with the following (triplet) weights:
\begin{eqnarray}
W^{(a)}_{bc}=\left\{
\begin{array}{cc} \omega_{ab}^{p_{ab}-1/2}\omega_{ac}^{p_{ac}-1/2}, & p_{ab}+p_{ac}=1\\
0,& \mbox{otherwise}.\end{array}
 \right.
 \label{DM-W}
\end{eqnarray}
Restating it in words, for the new representation an allowed (nonzero) configuration at any vertex corresponds to ones on  adjacent edges (there are two of them), of which one is necessarily ``active", i.e. present in the perfect matching and thus taking $\pi_{ab}=1$ value. After collecting the $w_{ab}$ pre-factors on the rhs of Eq.~(\ref{tf_DM}) into one
common multiplier, one observes that this modified model is in fact a particular Wick model (\ref{Z_WBG}). The
factor functions of the model is identical zero for configurations with at least two adjacent edges taking ``active"
$\pi=1$ value. Notice also that the WBG representation for the dimer model is
not unique as it can be build around any valid dimer configurations. Evidently any two valid
representations, described by the gauges $G_1$ and $G_2$ can be transformed to each other via
respective gauge transformations, $G_2 G^{-1}_1$.

We conclude this Subsection mentioning briefly yet another alternative reduction of the dimer model
to a WBG model, utilizing a topological transformation (specifically a contraction) of the original
graph ${\cal G}$. Consider a valid dimer configuration ${\bm p}_{0}$ as a set of edges
$\{a,b\}\in{\cal G}_{1}$ and contract each edge $\{a,b\}\in{\bm p}_{0}$ to a point. This results in
a new graph $\bar{{\cal G}}$, whose nodes $\bar{{\cal G}}_{0}={\bm p}_{0}$ are represented by the
edges that belong to the reference dimer configurations, and whose edges $\bar{{\cal G}}_{1}={\cal
G}_{1}\setminus{\bm p}_{0}$ are represented by the rest of the edges of ${\cal G}$. Each valid
dimer configuration ${\bm p}$ on ${\cal G}$ generates a configuration ${\bm\pi}$ of a binary edge
model on $\bar{{\cal G}}$ that has a form $\pi_{\alpha}=1$ if $\alpha\in{\bm p}$, and
$\pi_{\alpha}=0$, otherwise. (Here we consider an edge $\alpha\in \bar{{\cal G}}_{1}={\cal
G}_{1}\setminus{\bm p}_{0}\subset {\cal G}_{1}$ of $\bar{{\cal G}}$ as an edge of ${\cal G}$.) It
is easy to see that a vertex function $f_{a}({\bm\pi}_{a})$ for $a\in\bar{{\cal G}}_{0}$ is nonzero
only when all components of ${\bm\pi}_{a}$ are zeros (the case $a\in{\bm p}_{0}$), or when exactly
two components have the value one (otherwise), and also that the local weights satisfy the
conditions of the Wick binary model expressed in Eq.~(\ref{W-expressions}).

\subsection{Example: Ice Model for graphs of degree three}
\label{subsec:ice}

The ice model is defined in terms of orientations of edges on the graph. A valid
configuration/orientation is such that no vertex has adjusted edges all oriented onwards or
inwards. Here we consider the case when all vertices on ${\cal G}$ are of degree three. (Without
loss of generality our consideration here can be extended to vertices of degree not larger than
three.) Weight of any allowed orientation of the graph is unity. This model, $\#PL-3-NAE-ICE$
according to the complexity theory classification, was discussed in \cite{02Val,08Val} in lieu of
its reduction to dimer model via a holographic algorithm. Related models were studied extensively
in the mathematical physics literature \cite{67Lie_a,67Lie_b,07Bax}.

The ice model can be conveniently restated in terms of normal binary variables \footnote{An
alternative derivation, not requiring a graphical transformation, is discussed in Appendix
\ref{sec:gauge-extend}.}. This will require a graphical transformation consisting in breaking every
edge, $\{a,b\}$, by inserting a new auxiliary vertex $a-b$. Then, binary variables assigned to the
pair of new edges are $\pi_{a,a-b}=0, \pi_{b,a-b}=1$ if the direction of the arrow (for the
original orientation of the edge) was $a\to b$ and $\pi_{a,a-b}=1, \pi_{b,a-b}=0$ for the $b\to a$
arrow/orientation. The partition function of the respective binary model becomes
\begin{eqnarray}
&& Z_{ice}=\sum_{{\bm\pi}'}\left(\prod_{a\in{\cal G}_0}
f_a({\bm\pi}_a)\right)\left(\prod_{\{a,b\}\in{\cal G}_1} g_{a-b}(\pi_{a,a-b},\pi_{b,a-b})\right),
\label{ice1}\\
&& f_a({\bm\pi}'_a)=\left\{
\begin{array}{cc} 1, & \exists\  b,c\sim a,\quad \mbox{s.t.}\quad
\pi_{a,a-b}\neq\pi_{a,a-c}\\
0,&\mbox{otherwise}\end{array}\right.,\label{ice2}\\
&& g_{a-b}({\bm\pi}'_a)=\left\{
\begin{array}{cc} 1 & \pi_{a,a-b}\neq\pi_{b,a-b}\\
0,&\mbox{otherwise}\end{array}\right.,\label{ice2}
\end{eqnarray}
where
${\bm\pi}'={\bm\pi}\cup(\pi_{a,a-b}=0,1|\{a,b\}\in{\cal G}_1)$ and
${\bm\pi}'_a=(\pi_{a,a-b}=0,1|b\in \delta_{\cal G}(a))$.

Let us introduce the following gauge transformation on all edges of the newly formed extended graph
\begin{eqnarray}
G^{(ice)}_{a,a-b}=\frac{1}{\sqrt{2}}\left(\begin{array}{cc} 1 & 1\\ -1 & 1\end{array}\right),
\label{g-ice}
\end{eqnarray}
Factor functions of the ice model in the new basis becomes
\begin{eqnarray}
&& \tilde{f}_a(\pi_{a,a-1},\pi_{a,a-2},\pi_{a,a-3})=\frac{3}{\sqrt{2}}*\left\{
\begin{array}{cc} 1, & \pi_{a,a-1}=\pi_{a,a-2}=\pi_{a,a-3}=0\\
-1/3, & \sum_i\pi_{a,a-i}=2\\
0,&\mbox{otherwise}\end{array}\right.,\label{ice2a}\\
&& \tilde{g}_{a-b}({\bm\pi}'_a)=\left\{
\begin{array}{cc} 1, & \pi_{a,a-b}=\pi_{b,a-b}=0\\
-1, & \pi_{a,a-b}=\pi_{b,a-b}=1\\
0,&\mbox{otherwise}\end{array}\right.,
\label{ice3a}
\end{eqnarray}
where $(a,a-i)$ with $i=1,2,3$ marks three edges adjusted to vertex $a$. Accumulating the product
of the overall $3/\sqrt{2}$ terms into a common (normalization) multiplier, one immediately finds
that the resulting model is of the Wick's type (\ref{Z_WBG}) where
\begin{eqnarray}
\forall \alpha\neq \beta,\quad \alpha,\beta=a-1,a-2,a-3:\quad W^{(a)}_{\alpha,\beta}=-1/3; \quad \forall a\sim b:\quad W_{ab}^{(a-b)}=-1.
\label{W-ice}
\end{eqnarray}
Indeed, given that $\forall a\in{\cal
G}$ $|\delta_a({\cal G})|<4$, one simply does not have any $W^{(a)}$ terms in the ice model version
of (\ref{Z_WBG}) with $|\delta_a({\cal G})|>4$, and the only nontrivial vertices are these with
$|\delta_a({\cal G})|=2$. Notice, that this Wick feature would not apply to the ice model on planar
graphs with vertices of degree higher than three --- one can still use the gauge transformation but
the resulting transformed model will not maintain Wick's property.

Note that transformation of the ice model, identical to the one discussed above (e.g. utilizing
some alternative terminology) was also used as a show case of the  holographic algorithms of
Valiant in \cite{02Val,04Val,08Val}.

\subsection{Example: Ising Model}
\label{subsec:Ising}

We consider the classical Ising model without magnetic field:
\begin{eqnarray}
Z_I=\sum_{\bm\sigma}\exp\left(\sum_{\{a,b\}\in{\cal G}_0}\sigma_a J_{ab}\sigma_b\right),
\label{Ising-original}
\end{eqnarray}
where ${\bm\sigma}=(\sigma_a=\pm 1|a\in{\cal G}_0)$.
The original formulation (\ref{Ising-original}) is in terms of
binary variables associated with vertices, however the following EBG representation is more
appropriate for our purposes
\begin{eqnarray}
&& Z_{I}=\sum_{\bm\pi'}\left(\prod_{a\in{\cal G}_0} f_a({\bm\pi})\right)
\left(\prod_{\{a,b\}\in{\cal G}_1}g_{ab}(\pi_{a,a-b},\pi_{b,a-b})\right),\label{Is2}\\
&& f_a({\bm\pi})=\left\{\begin{array}{cc} 1, & \forall b,c\in\delta_{\cal G}(a)\quad \pi_{a,a-b}=\pi_{a,a-c},\\
0, & \mbox{otherwise}\end{array}\right.,\label{Is2}\\
&& g_{ab}(\pi_{a,a-b},\pi_{b,a-b})=\left\{\begin{array}{cc}
\gamma_{ab}, & \pi_{a,a-b}=\pi_{b,a-b}\\
\mu_{ab}, & \pi_{a,a-b}\neq\pi_{b,a-b}\end{array}\right.,
\label{Is3}
\end{eqnarray}
where $\gamma_{ab}=\exp(J_{ab})$,  $\mu_{ab}=\exp(-J_{ab})$ and we adopted notations of the previous Subsection for new (auxiliary) vertices breaking each edge of the original edge in two.

\begin{figure}[t]
\centering {\includegraphics[width=0.4\textwidth,page=20]{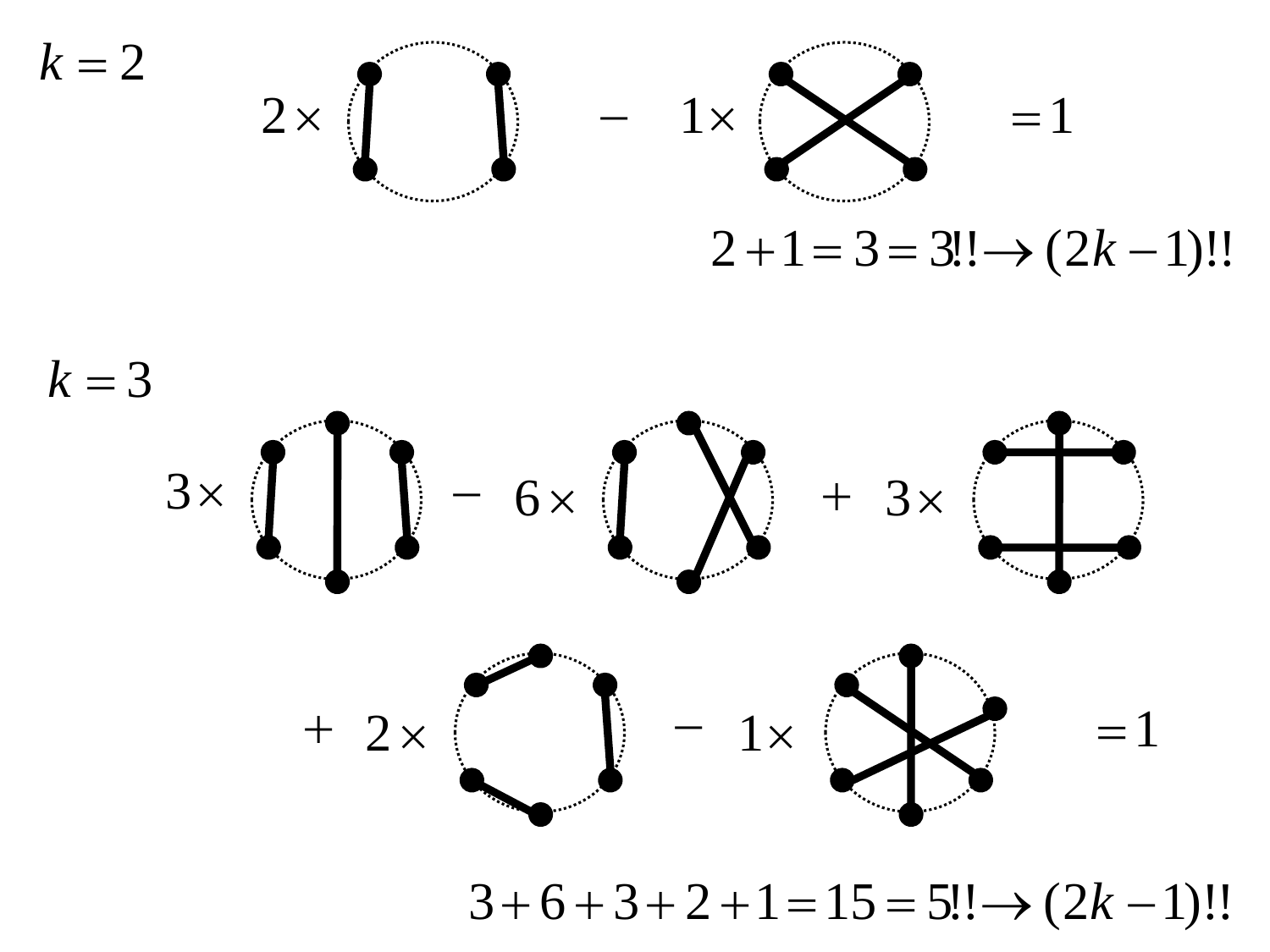}} \caption{Illustration
clarifying the Wick-feature of Eq.~(\ref{Is4}). \label{fig:Ising}}
\end{figure}

Applying the Gauge transformation (\ref{g-ice}), introduced above for the Ice model, to the Ising model one
arrives at the following factor functions in the new basis
\begin{eqnarray}
&& \tilde{f}_a({\bm\pi})=2^{-|\delta_{\cal G}(a)|/2 +1}*\left\{\begin{array}{cc} 1,
& \sum_{b\sim a}\pi_{a,b-a}=0\ \mbox{ mod }2,\\
0, & \mbox{otherwise}\end{array}\right.,\label{Is4}\\
&& \tilde{g}_{ab}(\pi_{a,a-b},\pi_{b,a-b})=(\gamma_{ab}+\mu_{ab})*\left\{\begin{array}{cc}
1, & \pi_{a,a-b}=\pi_{b,a-b}=0\\
(\gamma_{ab}-\mu_{ab})/(\gamma_{ab}+\mu_{ab}), & \pi_{a,a-b}=\pi_{b,a-b}=1\\
0, & \mbox{otherwise}\end{array}\right..
\label{Is5}
\end{eqnarray}
With the product of all the pre-factor terms, $2^{-|\delta_{\cal G}(a)|/2+1}$ and  $(\gamma_{ab}+\mu_{ab})$,
accounted for in the overall multiplier, the resulting model is one of the WBG models of
(\ref{Z_WBG}). This Wick feature of the modified graphical model is obvious for degree two vertices
of Eq.~(\ref{Is5}), and the only part which requires additional clarification concerns the relation
between a configuration of Eq.~(\ref{Is4}) with $\sum_{b\sim a}\pi_{a,b-a}>2$ and its pairwise
decomposition. Indeed, for a $\mathbb{Z}_{2}$-cycle  $\gamma$ and a vertex
$a\in\gamma_{0}$ of valence $|\delta_{\gamma}(a)|=2k$ with respect to $\gamma$, there are $(2k-1)!!$
local partitions with $\left((2k-1)!!+1\right)/2$ and $\left((2k-1)!!-1\right)/2$ of them having
modulo two intersection index equal to zero and one, respectively. Therefore according to the Wick
rules, the cumulative local contribution associated with the vertex $a$ is,
$\left((2k-1)!!+1\right)/2-\left((2k-1)!!-1\right)/2=1$, that is fully consistent with
Eq.~(\ref{Is4}). This simple combinatorial relation is illustrated in Fig.~\ref{fig:Ising} for
$k=2$ and $k=3$. Summarizing,  the Ising model is reduced to the WBG model with the following characteristics (W-weights):
\begin{eqnarray}
\forall b\neq c, b,c\sim a:\quad W^{(a)}_{b-a,c-a}=0;\quad \forall a\neq b:\quad W^{(a-b)}_{a,b}=\frac{\gamma_{ab}-\mu_{ab}}{\gamma_{ab}+\mu_{ab}}.
\label{WBG-Ising}
\end{eqnarray}

\subsection{$\#X$-matching model}
\label{subsec:X-matching}

The $\#X$-matching model is defined as follows \cite{02Val,04Val,08Val}:
On a planar bipartite graph, where $V_1$ and $V_2$ are bi-partitions of the full set of nodes and the nodes in $V_1$ all have degree $2$, consider matchings (monomer-dimer configurations) weighted in a way that all dimers (edges) have arbitrary weights, monomers on $V_1$ are all of unit weight and the weights of the monomers on $V_2$ are equal to minus the sum of the dimer weights of edges emanated from the vertex.

It is convenient to view the bipartite graph in the original definition of the $\# X$-matching model as a graph with the sets of nodes $V$ and edges $E$ represented by $V=V_{2}$ and $E=V_{1}$, respectively. We will adopt notations similar to these used in Section \ref{subsec:ice}, devoted to the Ice model, calling $a,b,\cdots$ the vertices of $V=V_2$ and $a-b,\cdots$ vertices belonging to $E=V_1$. Then the partition function $Z_{\# X}$ of the $\#X$-matching model has the following EBG representation on the graph ${\cal G}=({\cal G}_{0},{\cal G}_{1})=(V,E)$
\begin{eqnarray}
&& Z_{\# X}=\sum_{\bm\pi}\left(\prod_{a\in{\cal G}_0} f_a(\bm{\pi}_{a})\right)
\left(\prod_{\{a,b\}\in{\cal G}_1}g_{ab}(\pi_{a,a-b},\pi_{b,a-b})\right),\label{Z-X-1}\\
&& f_a({\bm\pi}_{a})=\left\{\begin{array}{cc} -\sum_{b\sim a}w_{ab}, & \forall b\sim a \quad \pi_{a,a-b}=0,\\
\sum_{b\sim a}w_{ab}\pi_{a,a-b}, & \sum_{b\sim a}\pi_{a,a-b}=1 \\
0, & \mbox{otherwise}\end{array}\right.,\label{Z-X-2}\\
&& g_{ab}(\pi_{a,a-b},\pi_{b,a-b})=\left\{\begin{array}{cc}
0, & \pi_{a,a-b}=\pi_{b,a-b}=1\\
1, & \pi_{a,a-b}=0 \;\;\; {\rm or}\;\;\; \pi_{b,a-b}=0 \end{array}\right.,
\label{Z-X-3}
\end{eqnarray}

Our task in this Subsection is to show via a gauge-transformation approach,  alternative to the match-gate approach of \cite{02Val,04Val,08Val}, that the partition function of the $\#X$-matching model, defined as the weighted number of the monomer-dimer configurations, is det-easy. To achieve this goal we will need to introduce, as in the above Subsection, a gauge transformation reducing the model to another model known to be det-easy.  However, the gauge transformation of this Subsection is principally different from these discussed in the rest of the paper, as here we will need to consider generalization from binary-setting to a $q$-ary, with $q=3$, alphabet. For this purpose we restate Eqs.~(\ref{Z-X-1}-\ref{Z-X-3}) as the following
$3$-ary Edge Graphical Model over the set of variables $\tilde{{\bm\pi}}=\{\tilde{\pi}_{a,a-b}\}$ that attain the values $\tilde{\pi}_{a,a-b}=0,\pm 1$:
\begin{eqnarray}
&& Z_{\# X}=\sum_{\tilde{\bm\pi}}\left(\prod_{a\in{\cal G}_0} \bar{f}_a(\tilde{\bm{\pi}}_{a})\right)
\left(\prod_{\{a,b\}\in{\cal G}_1}\bar{g}_{ab}(\tilde{\pi}_{a,a-b},\tilde{\pi}_{b,a-b})\right),\label{Z-HS-1}\\
&& \bar{f}_a(\tilde{\bm\pi}_{a})=\left\{\begin{array}{cc} 1, & \forall b\sim a \quad \tilde{\pi}_{a,a-b}=0,\\
w_{ab}\tilde{\pi}_{a,a-b}, & \tilde{\pi}_{a,a-b}=\pm 1 \;\;\; {\rm and} \;\;\; \forall c\ne b\;\; \tilde{\pi}_{a,a-c}=0 \\
0, & \mbox{otherwise}\end{array}\right.,\label{Z-HS-2}\\
&& \bar{g}_{ab}(\tilde{\pi}_{a,a-b},\tilde{\pi}_{b,a-b})=\left\{\begin{array}{cc}
0, & \tilde{\pi}_{a,a-b}=\tilde{\pi}_{b,a-b}=1\\
1, & \mbox{otherwise} \end{array}\right..
\label{Z-HS-3}
\end{eqnarray}
To see equivalence between the binary and $q$-ary representations let us walk backwards, from  Eqs.~(\ref{Z-X-1}-\ref{Z-X-3}) to Eqs.~(\ref{Z-HS-1}-\ref{Z-HS-3}), and associate
with any configuration $\tilde{\bm\pi}$ of our $3$-ary model a binary configuration ${\bm\pi}$ by setting $\pi_{a,a-b}=1$ when $\tilde{\pi}_{a,a-b}=1$, and $\pi_{a,a-b}=0$ when $\tilde{\pi}_{a,a-b}=0,-1$. Then, we do partial summation over all $\tilde{\pi}_{a,a-b}=0,-1$ correspondent to $\pi_{a,a-b}=0$. In a general case the result of such partial summation would be non-local with respect to the graph. However, in our special case the resulting model happens to be graph-local. Indeed, the weights $\bar{g}_{ab}(\tilde{\pi}_{a,a-b},\tilde{\pi}_{b,a-b})$ in Eq.~(\ref{Z-HS-3}) actually depend on the reduced variables $\pi_{a,a-b}$ and $\pi_{b,a-b}$ and the partial summation can be performed for each node $a\in{\cal G}_{0}$ independently. Finally, we arrive at the binary formulation of Eqs.~(\ref{Z-X-1}-\ref{Z-X-3}).

One may wonder why changing from binary to $3$-ary representation is useful? In fact it is very useful because of a greater freedom of further transformation the $3$-ary representation offers. Indeed, we will see below that relating $\tilde{\pi}$ variables to $\pi$  variables in a new way allows to reduce the $3$-ary model to the already known det-easy model (specifically, a dimer model, already shown reducible to a WBG model in Section \ref{subsec:dimer-example}).

Namely, we set $\pi_{a,a-b}=1$ when $\tilde{\pi}_{a,a-b}=\pm 1$, and $\pi_{a,a-b}=0$ when $\tilde{\pi}_{a,a-b}=0$. The resulting binary model turns out to be the dimer model on the reduced version of ${\cal G}$ (no half-edges only full edges) with the dimer weights $-w_{ab}w_{ba}$ for the $\{a,b\}$ edges. This follows from the following arguments. Consider a link $\{a,b\}$. If $\pi_{a,a-b}=1$ and $\pi_{b,a-b}=0$ we have $\tilde{\pi}_{a,a-c}=0$ for all $c\sim a$ with $c\ne b$, and also $\bar{g}_{ab}(\tilde{\pi}_{a,a-b},\tilde{\pi}_{b,a-b})=1$, so that the partial summation over $\tilde{\pi}_{a,a-b}=\pm 1$ can be performed independently, which results in $0$. If $\pi_{a,a-b}=\pi_{b,a-b}=1$, we have $\tilde{\pi}_{a,a-c}=0$ for all $c\sim a$ with $c\ne b$, and also $\tilde{\pi}_{b,b-c}=0$ for all $c\sim b$ with $c\ne a$. This allows the summations over $\pi_{a,a-b}=\pm 1$ and $\pi_{b,a-b}=\pm 1$ to be performed independently, thus resulting in the factors of $0$ for $\pi_{a,a-b}=\pi_{b,a-b}=1$, and $w_{ab}w_{ba}\pi_{a,a-b}\pi_{b,a-b}$, otherwise.  QED.

It is worth to note that in terms of the language of loop towers for $q$-ary models, developed in our earlier work on loop calculus for $q$-ary alphabets \cite{07CC}, the partial re-summation based on partitioning the alphabet into one special letter ($\tilde{\pi}=0$) and the rest ($\tilde{\pi}=\pm 1$) should be viewed as a partial gauge fixing that corresponds to the lowest level of the loop tower.

\section{Discussion and Future Plans}
\label{sec:concl}

In the present manuscript we have identified a broad class of models that allow for an easy
solution on planar graphs. Our approach is based on the gauge invariance of vertex models on
graphs that has been implemented in our previous work in the context of Belief Propagation and Loop
Series and the notion of a Wick vertex model, referred to as WBG. The WBG models are described in
terms of an arbitrary choice of weights associated with edges of the graph, and explicit construction of the high-degree (even coloring) vertex weights from the respective edge weights according to the so-called Wick rules. We have established an equivalence
between the WBG models and the Gaussian fermion (Grassmann variables based)  models on the same graph,
referred to as the VG$^3$models (Theorem~\ref{Theorem}).  The models are ``easy" as their partition
functions are equal to Pfaffians of square matrices of the size equal to twice the number of edges of the graph.
We have also demonstrated that the dimer model that constitutes the ``solving tool'' of the holographic algorithm of Valiant
\cite{02Val,04Val,08Val} is gauge equivalent to a WBG model. The other two components of the
holographic approach include graphical transformations (extensions of the original graph) and
linear transformations, which can be described as particular cases of gauge transformations using
the language adopted in this manuscript. If one represents the partition function of the dimer
model on the extended graph, obtained by using the gadgets of the holographic algorithm, as a
Gaussian Grassmann integral and integrate over the variables related to the new vertices that
result from the graph extension, the result will be represented by a Gaussian Grassman integral on
the original graph. Therefore, our approach provides an explicit and invariant description of the
models that can be treated using the holographic algorithm as the models gauge equivalent to WBG
models. It also provides with an alternative way to obtain an easy solution, while sticking to the
original graph all the time.

Some problems we plan to address in the future, extending on the results/approach described in the
manuscript, are:
\begin{itemize}
    \item Generalize the approach to account for det-easy planar graphical models defined over $q$-ary alphabet with $q>2$. As illustrated in Subsection \ref{subsec:X-matching}, there exist an interesting new class of more general $q$-ary models which are reducible to det-easy binary graphical models.

    \item Use the described hierarchy of easy planar models as a basis for efficient variational
    approximation of planar but difficult problems, and possibly more generally non-planar models. Then, on the next level, build a controlled perturbative corrections to the variational result. Notice, that this approach may also be useful for building efficient variational matrix-product state wave functions for quantum planar models, e.g. in the spirit of \cite{06VC}.

    \item Analysis of Wick Gaussian models on surface graphs of nonzero genus. This work will extend the classical results \cite{63Kas,96RZ,99GLa,99GLb,00RZ,07CR,08CR} stating that partition function of dimer models on graphs embedded in surface of genus $g$, so-called surface graphs, is expressed as a sum over $2^{2g}$ Pfaffians, each correspondent to a ``spinor" structure parameterizing the equivalence classes of Kasteleyn orientations on the locally planar graphs. The approach developed in this manuscript allows for a relatively straightforward extension to the case of surface graphs. In particular, we plan to show in a forthcoming publication that Kasteleyn orientations on the extended graph ${\cal G}^e$ generate all $2^{2g}$ spinor structures on the embedding Riemann surface.

    \item Study Wick Gaussian models on non-planar but Pfaffian orientable or $k$-Pfaffian orientable graphs (thus any dimer model on surface graph of genus $g$ is $2^{2g}$-Pfaffian orientable),  also utilizing the new results from graph theory on the Pfaffian orientability \cite{06Tho}.

    \item For the case of a general EBG model we will apply the statistical variational principle to build the ``best'' approximation in a form of a WBG model and address the problem of finding efficient ways to account for the corrections.
\end{itemize}

\section{Acknowledgments}

We are thankful to David Gamarnik for attracting our attention to the holographic algorithms of \cite{02Val,04Val,08Val}, to John R. Klein, Jason Johnson and Vicenc Gomez for useful discussions and comments, and also to an anonymous Referee whose comments force us to reconsider (and hopefully improve) the presentation sequence. This material is based upon work supported by the National Science Foundation under
CHE-0808910 (VC) and CCF-0829945 (MC via NMC). The work at LANL was carried out under the auspices of the National Nuclear Security Administration of the U.S. Department of Energy at Los Alamos National Laboratory under Contract No. DE-AC52-06NA25396. MC also acknowledges partial support of KITP at UCSB where some part of this work was done.

\bibliographystyle{apsrev}
\bibliography{Planar,Gauss,BP_review}

\appendix

\section{Topology behind the door: Self-intersection invariant of immersions and spinor structures}
\label{sec:topology}

The main result of this manuscript is represented by Theorem~\ref{Theorem}. The proof of the
Theorem rests on a computation of the relevant Gaussian Grassmann integral using a certain
expansion, combined with the Wick Theorem and Lemma~\ref{Lemma}.  In this
Appendix we present a brief qualitative discussion (without proofs) of the topology that stands
behind the presented combinatorial proofs presented in the main body of the manuscript. This will connect our results to the approach developed
in \cite{07CR,08CR} for the dimer model on surface graphs, and also provide some ideas on how
our results on the equivalence between the Wick and binary models can be extended to the surface
graph case (the details of this description and related proofs will be published elsewhere).

A topological nature of the equivalence between the WBG and G$^3$ models can be seen by comparing
Eqs.~(\ref{Z-W-expand}) and (\ref{Z-F-expand}) that play a major role in the proof of
Theorem~\ref{Theorem}. Each contribution in both sums can be represented as a product of the
individual contributions labeled by the immersed orbits $C_{j}(\gamma,{\bm\xi})$. The individual
contribution associated with an immersed orbit $C$ has the same absolute value for both
expansions, they differ by the sign factors only. For the G$^3$ model the sign associated with an
immersed orbit $C$ is $-\varepsilon(C)=(-1)^{s(C)+1}$ (where the last expression defines $s(C)$ as a function of $\varepsilon(C)$ introduced in the main text), in the WBG model case the sign is
$(-1)^{N(C)}$ with $N(C)$ being the number of self-intersections of $C$. Note that the
intersections between different orbits in the expansion of $\gamma$, although apparently affecting
the sign factor in Eq.~(\ref{Z-W-expand}), can be actually eliminated from the consideration, since
in the planar case two orbits always intersect at an even number of points. In this appendix we
will argue that $s(C)$ can be associated with the unique spinor structure on $\mathbb{R}^{2}$, and
the equivalence between the models originate from the topological formula, $N(C)+1=s(C) \; ({\rm
mod}\; 2)$, that relates the self-intersection number to the spinor structure.

We proceed with more precise formulations. Note that the statement of Lemma~\ref{Lemma} can be
interpreted in the following way: The function $q(\gamma,{\bm\xi})$ of $\gamma$ and ${\bm\xi}$
defined by Eq.~(\ref{define-q}) can be viewed as a $\mathbb{Z}_{2}$-valued function defined on the
sets of immersed orbits, represented by $C_{j}(\gamma,{\bm\xi})$, and it satisfies the property
$q(\gamma,{\bm\xi})=0$.
$q(\gamma,{\bm\xi})$ consists of three contributions: the number
$N(\gamma,{\bm\xi})$ of intersections and self-intersections of the involved orbits (naturally
modulo two), the number $n(\gamma,{\bm\xi})$ of the orbits involved, and the contributions
$s\left(C_{j}(\gamma,{\bm\xi})\right)$ that arise from the individual orbits, where $s\equiv (\log(\varepsilon)/(i\pi)$.
What we intend to illustrate is that by combining the modulo two numbers of
intersections and the modulo two number of orbits in a decomposition of a planar
$\mathbb{Z}_{2}$-cycle $\gamma$ we arrive at a topological invariant
$p(\gamma,{\xi})=N(\gamma,{\bm\xi})+n(\gamma,{\bm\xi})\in\mathbb{Z}_{2}$, referred to as the
self-intersection invariant, whereas the individual orbit contributions
$s\left(C_{j}(\gamma,{\bm\xi})\right)$ and the corresponding sign factors $\varepsilon=(-1)^{s}$
are associated with a spinor structure on the plane $\mathbb{R}^{2}$, the original planar graph and
its extension are imbedded into. Combining the self-intersection invariant with the spinor
structure related contributions we arrive at the $\mathbb{Z}_{2}$-invariant $q=p+s$ that depend on
the planar $\mathbb{Z}_{2}$-cycle $\gamma$ only, and $q=0$, since all cycles in $\mathbb{R}^{2}$
are contractable. This will be done by interpreting the orbits as closed trajectories of a free
particle in the plane, also referred to as a free planar particle.

\subsection{Immersions, self-intersection invariant and spinor structures on $\mathbb{R}^2$}
\label{subsec:spin-immersions}

\begin{figure}
\centering \subfigure[An example of a smooth closed trajectory. Parametrization chosen for plotting
is ${\bm r}=(r_1(t),r_2(t))=(\sin(t),\sin(2t))$, where $t$ changes from $0$ to
$2\pi$.]{\includegraphics[width=0.25\textwidth,page=1]{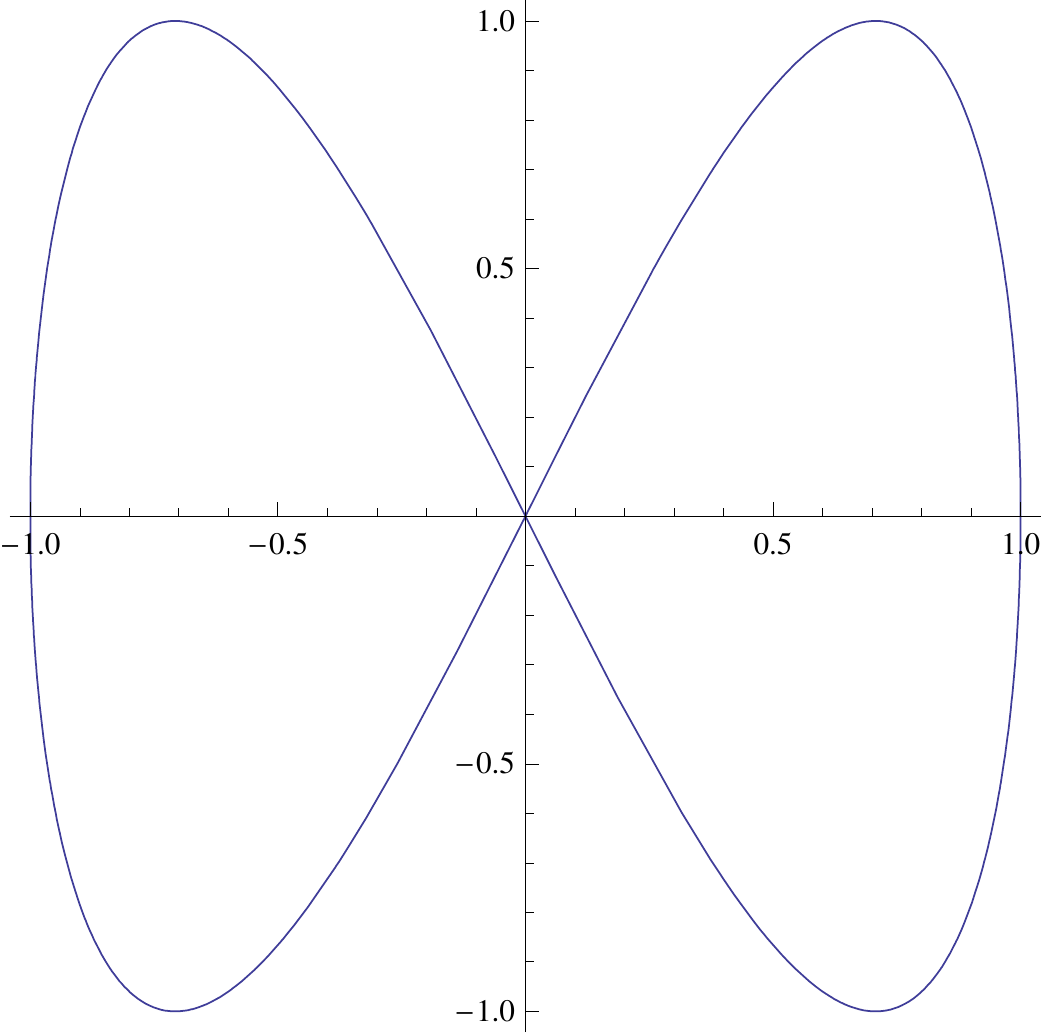}} \subfigure[Direction of velocity
(angle), $\varphi(t)=\arctan(2\cos(2t)/\cos(t))$, for a particle moving along the immersed orbit
from (A), thus ${\bm r}$ are coordinates of the particle in $\mathbb{R}^{2}$, as a function of
``time" $t$. The invariant, measuring total number of angular rotations, is $0$ for this
example.]{\includegraphics[width=0.35\textwidth,page=2]{lisagu.pdf}} \subfigure[Trajectory of the
particle in the phase-space, ${\bm x}=(r_1(t),r_2(t),\varphi(t))\in M^{3}=\mathbb{R}^{2}\times
S^{1}$.]{\includegraphics[width=0.15\textwidth,page=3]{lisagu.pdf}} \caption{Illustration of an
immersion in $\mathbb{R}^{2}$ of an orbit with one crossing and its phase-space (particle)
counterpart.\label{fig:lisagu.pdf}}
\end{figure}

This Subsection introduces some useful objects that are new to the manuscript, as defined for
continuous spaces, e.g. $\mathbb{R}^2$, rather than for the graphical structure as in the main part
of the manuscript. This continuous spaces approach will be used in the following Subsection for a
topological illustration of Lemma \ref{Lemma}. The arguments presented there can be easily
converted into a proof, however, this task goes beyond the scope of this manuscript. We will start
with introducing a special type of smooth maps (trajectories) $C:S^{1} \to \mathbb{R}^{2}$,
referred to as {\it immersions} and denoted $C:S^{1} \looparrowright \mathbb{R}^{2}$. The main
reason for bringing the immersions to our discussion is that the immersed closed walks, that
represent the immersed orbits $C_{j}(\gamma,{\bm\xi})$ in the decomposition of a
$\mathbb{Z}_{2}$-cycle $\gamma$ on a planar graph ${\cal G}\subset \mathbb{R}^{2}$ in
Lemma~\ref{Lemma}, can be deformed in a natural way to immersions $\bar{C}_{j}:S^{1}
\looparrowright \mathbb{R}^{2}$ and studied using topological methods. This rationalizes the terms
{\it immersed closed walk} and {\it immersed orbit}, we have introduced in the context of graphs.
We further discuss the self-intersection invariant $p:{\rm Imm}(S^{1},\mathbb{R}^{2}) \to
\mathbb{Z}_{2}$ that associates with an immersion $C:S^{1} \looparrowright \mathbb{R}^{2}$ the
number of self-intersections modulo two. The self-intersection invariant provides a topological
interpretation of the intersection number $N(\gamma,{\bm\xi})$ that appears in Lemma~\ref{Lemma}.
In the context of the self-intersection invariant we will refer to Smale-Hirsch-Gromov (SHG) theorem \cite{71Hae} that
reduces the problem of studying the topological properties of immersions $C:S^{1} \looparrowright
\mathbb{R}^{2}$ to a much simpler problem of studying their phase-space counterparts $f(C):S^{1} \to
\mathbb{R}^{2}\times S^{1}$. In particular, the SHG theorem
 rationalizes a natural
appearance of the phase space $M^{3}=\mathbb{R}^{2}\times S^{1}$. We proceed with introducing the
notion of the spinor structure on $\mathbb{R}^{2}$ viewed as a binary function $s:LM^{3}\to
\mathbb{Z}_{2}$ that associates with any phase-space trajectory $\tilde{C}:S^{1}\to M^{3}$ a number
$s(\tilde{C})\in\mathbb{Z}_{2}$ and can be also defined for the immersions $C:S^{1} \looparrowright
\mathbb{R}^{2}$ by $s(C)=s(f(C))$. The spinor structure provides a topological interpretation of
the factors
$\varepsilon\left(C_{j}(\gamma,{\bm\xi})\right)=(-1)^{s\left(\bar{C}_{j}(\gamma,{\bm\xi})\right)}$
in terms of the value of the spinor structure on the relevant immersions. Finally, we briefly
discuss a topological relation between the spinor structure and self-intersection invariant that
stands behind the statement of Lemma~\ref{Lemma}.

We start with introducing a concept of an immersion. To that end we consider a particle moving
smoothly over a closed trajectory in a plane two-dimensional space $\mathbb{R}^2$, where thus
smooth planar trajectory $C:S^{1}\to \mathbb{R}^{2}$ parameterized by ${\bm r}(t)=(r_1(t),r_2(t)$
where the ``time" $t$ belongs to $S^{1}$. The trajectory is called an {\it immersion}, which is
denoted by $C:S^{1} \looparrowright \mathbb{R}^{2}$ if the velocity stays non-zero, i.e., $\dot{\bm
r}(t) \ne 0$, at all $t$. In the following the key space will be the iso-energetic shell
$M^{3}=\mathbb{R}^{2}\times S^{1}$ of a free planar particle so that for ${\bm
x}=(r_1(t),r_2(t),\theta(t))\in \mathbb{R}^{2}\times S^{1}$, where the angular variable $\theta\in
S^{1}$ describes the velocity direction (i.e., the normalized particle velocity) $\dot{{\bm
r}}(t)\left|\dot{{\bm r}}(t)\right|^{-1}\in S^{1}$. Following the terminology, commonly accepted in
mathematical physics, we will also refer to the iso-energetic shell as the particle phase space.
For a smooth planar trajectory $C$ with a non-zero at all times velocity, i.e., an immersion, we
denote by $f(C)$ its phase-space counterpart with the normalized velocity, which clearly forms a
trajectory in $M^{3}$ (see Fig.~\ref{fig:lisagu.pdf} for an illustrative example).

Our next step, as outlined above, is to describe the self-intersection invariant, defined as the
modulo two number of self-intersections $i(C)=N(C)\; ({\rm mod}\; 2)$ of a smooth planar trajectory
$C:S^{1}\to \mathbb{R}^{2}$. It is intuitively clear that $i(C)$ is a topological invariant only
for immersions $S^{1}\looparrowright \mathbb{R}^{2}$, i.e., it stays the same only if we deform a
trajectory so that it stays an immersion all the time, since in this case the self-intersections
can appear and disappear in pairs only \footnote{If the condition of having everywhere non-zero
velocity is relaxed, this would be not true anymore. This can be illustrated using a simple
example. One can imagine a closed trajectory of the shape of ``figure eight'' from
Fig.~\ref{fig:lisagu.pdf}, being deformed in a way that one half of it becomes smaller turning into a
point with a cusp, followed by smothering the cusp, which results into a circle. The number of
self-intersections changes by one. Note that at the stage of the deformation when we have a cusp,
the velocity of the trajectory at a cusp turns to zero. This means that the map $C:S^{1}\to M^{2}$
at this stage is not an immersion, which actually allows the parity of the self-intersection number
to be changed.}. Immersions can be conveniently studied using the SHG theorem that can be
applied for our purposes in the following way. Interpret the procedure of associating with an
immersion $C$ its phase space counterpart $f(C)$, as described above, as a map $f:{\rm
Imm}(S^{1},\mathbb{R}^{2})\to LM^{3}$ from the space of immersion of a closed trajectory on a plane
with everywhere nonzero velocity to the space $LM^{3}$ of closed trajectories in
$M^{3}=\mathbb{R}^{2}\times S^{1}$. The SHG theorem, applied to our case, claims that $f$
is a weak homotopy equivalence. In particular it implies that given two immersions, $C,C'\in {\rm
Imm}(S^{1},\mathbb{R}^{2})$, to answer the question whether $C$ can be deformed to $C'$, while {\it
staying an immersion during the whole deformation}, is equivalent to finding out whether $f(C)$ can
be deformed to $f(C')$ {\it without any further restrictions}. Note that the SHG theorem
explicitly classifies the immersions $S^{1}\looparrowright \mathbb{R}^{2}$: since the coordinate
space $\mathbb{R}^{2}$ is contractible, only the velocity component $\theta_{C}:S^{1}\to S^{1}$ of
the phase-space trajectory $f(C):S^{1} \to \mathbb{R}^{2}\times S^{1}$ matters, with $\theta_{C}$
being fully homotopically determined by its degree ${\rm deg}(\theta_{C})\in\mathbb{Z}$ that
describes the total number of rotations the velocity direction makes over the whole trajectory. In
particular two immersions $C,C'\in {\rm Imm}(S^{1},\mathbb{R}^{2})$ are equivalent if and only if
${\rm deg}(\theta_{C})={\rm deg}(\theta_{C'})$, i.e. the direction of velocity (the angle) makes
the same number of rotations over both trajectories.

Consider concatenation of two trajectories $C_{1}\star C_{2}$ that share the starting point point
in $\mathbb{R}^2$, understood as trajectory of a particle cycling $C_2$ first and then going over
the $C_1$. The concatenation results in
\begin{eqnarray}
\label{concatenate-loops} i(C_{1}\star C_{2})=i(C_{1})+i(C_{2})+C_{1}\cdot C_{2}+1,
\end{eqnarray}
where $C_{1}\cdot C_{2}$ denotes the modulo two intersection index (the modulo two number of
intersection points) of $C_{1}$ and $C_{2}$. Eq.~(\ref{concatenate-loops}) reflects the fact that
the number of self-intersections of a concatenation of $C_{1}$ with $C_{2}$ is the number of
self-intersections of $C_{1}$ plus the number of self-intersections of $C_{2}$ plus the number of
intersections of $C_{1}$ with $C_{2}$ plus one (the self-intersection at the point of
concatenation), which, as opposed to the self-intersection index, is well-defined for any
continuous closed trajectories, not necessarily immersions. Note that for the planar case discussed
here $C_{1}\cdot C_{2}=0$ (intersections  appear and disappear in pairs), stated differently it
reflects the fact that all cycles in $\mathbb{R}^{2}$ are contractable. Nevertheless, we keep these
contributions, since for our application it will be convenient to treat intersections and
self-intersections on equal footing. By defining $p(C)=i(C)+1$ we arrive at
\begin{eqnarray}
\label{quadratic-form} p(C_{1}\star C_{2})=p(C_{1})+p(C_{2})+C_{1}\cdot C_{2}.
\end{eqnarray}

By the SHG theorem the topological classes of immersions $S^{1}\looparrowright
\mathbb{R}^{2}$ are fully described by the topological classes of their phase space counterparts
$f(C):S^{1}\to M^{3}\cong \mathbb{R}^{2}\times S^{1}$, the latter being labeled by the first
homology group $H_{1}(M^3)\cong \mathbb{Z}$, counting the number of rotations of the velocity
orientation along the immersion. The self-intersection invariant, therefore, is also defined in
terms of the natural factorization $p:\mathbb{Z}\to\mathbb{Z}_{2}$ (i.e., it can be viewed as a
composition $LM^{3}\to H_{1}(M^{3})\to \mathbb{Z}_{2}$ where the left map associates with a closed
phase space trajectory the number of the velocity rotations), and satisfies the property
\begin{eqnarray}
\label{quadratic-form-2} p(\gamma_{1}+\gamma_{2})=p(\gamma_{1})+p(\gamma_{2})+g(\gamma_{1})\cdot
g(\gamma_{2})=p(\gamma_{1})+p(\gamma_{2}),
\end{eqnarray}
where $g(\gamma)$ is a closed trajectory in $\mathbb{R}^{2}$ that is obtained by reduction
(projection) from $M^{3}$ by simply ignoring information about orientation of the particle
velocity. $g(\gamma_{1})\cdot g(\gamma_{2})$ stands for the intersection index (mod $2$) of the two
reduced closed trajectories, and for the planar $\mathbb{R}^{2}$ case considered in the paper it is
equal to zero, i.e., $g(\gamma_{1})\cdot g(\gamma_{2})=0$.
%In other words, for any immersion the self-intersection index (modulo two)
%is given by the number of rotations of the unit velocity along the trajectory plus one.

Consider a composite trajectory defined as a union of smooth closed and not necessarily
intersecting trajectories (immersions). In this case we define
$p\left(\sum_{j=1}^{n}C_{j}\right)=i\left(\bigsqcup_{j=1}^{n}C_{j}\right)+n$, which generalizes the
definition of a single smooth trajectory (an immersion). Note that this definition ensures the
property given by Eq.~(\ref{quadratic-form-2}). Summarizing, we have an invariant $p(\gamma)$
defined on $H_{1}(M^{3})\cong \mathbb{Z}$ that satisfies Eq.~(\ref{quadratic-form-2}) and if a
complex trajectory is decomposed into a set of smooth closed trajectories we derive
\begin{eqnarray}
\label{quadratic-form-3}
p(\gamma)=p\left(\bigsqcup_{j=1}^{n}C_{j}\right)%=p\left(\sum_{j=1}^{n}C_{j}\right)
=\sum_{j=1}^{n}i(C_{j})+\sum_{j<k}g(C_{j})\cdot
g(C_{k})+n=\sum_{j=1}^{n}p(C_{j})+\sum_{j<k}g(C_{j})\cdot g(C_{k})=\sum_{j=1}^{n}p(C_{j}).
\end{eqnarray}
This completes our discussion of the self-intersection invariant.

We are now in a position to introduce the notion of a spinor structure, as outlined in the
beginning if this subsection. A spinor structure can be interpreted as a map $s:LM^{3}\to
\mathbb{Z}_{2}$ that associates with any phase-space trajectory $\tilde{C}:S^{1} \to
\mathbb{R}^{2}\times S^{1}$ its parity $s(\tilde{C})=0,1$. The spinor structure map $s$ should
satisfy the following two conditions: (a) Any continuous deformation of $\tilde{C}$ does not change
$s(C)$, or more accurately (and generally) if $\tilde{C}\sqcup\tilde{C'}$ form a boundary of a two-dimensional
domain then $s(\tilde{C})=s(\tilde{C'})$; and (b) $s(\tilde{C}_{{\bm r}})=1$ for $\tilde{C}_{{\bm
r}}(t)=({\bm r},t)$ being a phase space trajectory with the constant position ${\bm r}$ and the
velocity making a full rotation. Obviously, the spinor structure can be also viewed as a map
$s:{\rm Imm}(S^{1},\mathbb{R}^{2})$ if we define $s(C)=s(f(C))$ for $C:S^{1}\looparrowright
\mathbb{R}^{2}$. For this interpretation the condition (b) means $s(C_{{\bm r}})=1$ where the
immersion $C_{{\bm r}}$ represents a circle on $\mathbb{R}^2$, centered at ${\bm r}$ with a very
small radius. Obviously, according to the Stokes theorem applied to the phase space, any spinor
structure $s(C)$ can be represented as
\begin{eqnarray}
\label{define-Dirac-phase} s(C)=s(f(C))=\int_{f(C)}A_{j}({\bm x})dx^{j}, \;\;\; \int_{\{{\bm
r}\}\times S^{1}}A_{j}({\bm x})dx^{j}=1,
\end{eqnarray}
where ${\bm A}$ is an Abelian curvature-free vector potential ${\bm A}$, in the phase space
$M^{3}=\mathbb{R}^{2}\times S^{1}$. The curvature-free condition reads
$F_{ij}=\partial_{i}A_{j}-\partial_{j}A_{i}=0$, thus expressing the condition (a) formulated above.
From the definitions of the spinor structure and the vector potential, one finds that the value
$s(C)$ of the spinor structure map on any immersion $C:S^{1}\looparrowright\mathbb{R}^2$ is given
by the modulo two total number of the particle velocity rotations along the closed trajectory (see
Fig.~\ref{fig:lisagu.pdf}c for an illustration). Note that for our planar case the conditions (a) and
(b) define a unique function $s$, which means that there is a unique spinor structure on
$\mathbb{R}^{2}$ (this is not true in the surface case). Also, this unique spinor structure can be
described by different gauge fields ${\bm A}$ that differ by gauge transformations $A_{j}\mapsto
A_{j}+\partial_{j}\phi$. This completes our discussion of the spinor structures.

\subsection{Topological interpretation of Lemma~\ref{Lemma}}
\label{subsec:lemma-top}

\begin{figure}[t]
\centering {\includegraphics[width=0.5\textwidth,page=21]{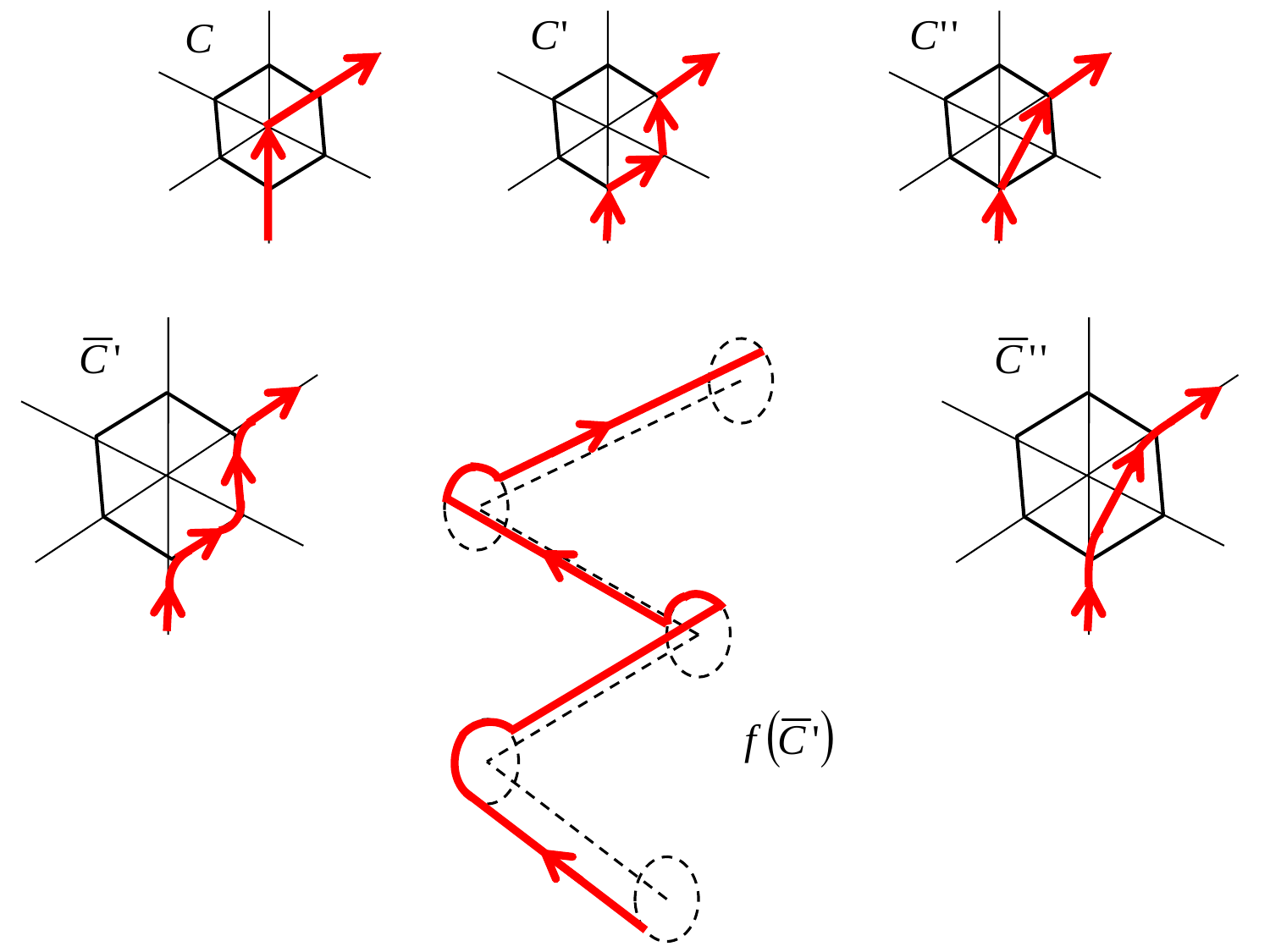}} \caption{Schematic illustrations
for details of the transformation procedure from an element of a directed orbit on a graph to its particle-trajectory smooth analog embedded into $\mathbb{R}^2$ and consequently into
$M^3$. Notice that if the direction of the orbit is reversed the trajectory $f(\bar{C}')$ in $M^3$ will change. Underlying spinor structure is not shown on the figure.
\label{fig:smoothing}}
\end{figure}

To demonstrate the equivalence of the combinatorial and topological definitions of $\varepsilon(C)$
we view a directed immersed orbit $C'$ on ${\cal G}^e$, associated with an immersed orbit on ${\cal G}$ as
a continuous piece-wise smooth trajectory in $\mathbb{R}^{2}$. See Fig.~\ref{fig:smoothing} for a schematic illustration of
the orbit-to-trajectory transformation.
In very small neighborhoods of the nodes
we deform $C'$ in an obvious way, so that we do not gain self-intersections in these small
neighborhoods, to obtain an immersion $\bar{C}':S^{1}\looparrowright \mathbb{R}^{2}$. In the small
neighborhoods of the nodes, where the deformations have been performed, the coordinate ${\bm r}$,
i.e., the $\mathbb{R}^{2}$ component of the corresponding phase-space trajectory
$f(\bar{C}'):S^{1}\to M^{3}$ is almost constant, whereas the velocity orientation $\theta$ makes a
finite rotation in the counterclockwise or clockwise directions for the left and right turns,
respectively. Following $f(\bar{C}')$ on the deformation of $\bar{C}'$ back to $C'$ we obtain a
continuous piece-wise smooth trajectory in $M^{3}$ that represents a orbit in $M^{3}$, denoted with
some minor abuse of notation $f(C'):S^{1}\to M^{3}$, whose smooth pieces $B_{ab}$ and
$B_{ac}^{(b)}$ correspond to the edges and the nodes (velocity turns at the nodes) of $C'$,
respectively. The value of $s\left(f(C')\right)$ can be considered as a sum of the contributions
from the smooth pieces, labeled by the edges and nodes of $C'$, according to
Eq.~(\ref{define-Dirac-phase}). Indeed, we denote by $B_{ab}$ a smooth path in $\mathbb{R}^{2}$
that represents the edge of ${\cal G}^e$ in the embedding ${\cal G}\subset \mathbb{R}^{2}$,
whereas for a triplet $a\to b\to c$ of ${\cal G}^e$ we denote by $B_{ac}^{(b)}$ the shortest path
on $\{{\bm r}_{b}\}\times S^{1}$ that connects the directions of the velocities at ${\bm r}_{b}$ on
the paths $B_{ab}$ and $B_{bc}$, respectively. Then the aforementioned edge and vertex
contributions have the form
\begin{eqnarray}
\label{define-s-local} s_{ab}=s(f(B_{ab}))=\int_{f(B_{ab})}A_{j}({\bm x})dx^{j}, \;\;\;
s_{ac}^{(b)}=\int_{B_{ac}^{(b)}}A_{j}({\bm x})dx^{j},
\end{eqnarray}
respectively. Naturally, $\varepsilon(f(C'))=(-1)^{s(f(C'))}$ is given by the product of the edge
and node contributions and, provided the latter are represented by the Kasteleyn edge orientations
and left triplet orientations, which can be always achieved by a gauge transformation of the field
${\bm A},$ we reproduce the original combinatorial form, i.e.,
$\varepsilon_{C}(C)=\varepsilon_{T}(f(C'))$, where the subscripts $C$ and $T$ stand for
combinatorial and topological, respectively.

To evaluate $q(\gamma)$ for a $\mathbb{Z}_{2}$ cycle we decompose it into $n(\gamma,{\bm\xi})$
immersed orbits $C_{j}=C_{j}(\gamma,{\bm\xi})$ and consider $C_{j}$ and $C''_{j}$ as continuous
piece-wise smooth orbits in $\mathbb{R}^{2}$. We denote by $\bar{C}_{j}:S^{1}\looparrowright
\mathbb{R}^{2}$ and $\bar{C}''_{j}:S^{1}\looparrowright \mathbb{R}^{2}$ the immersions, obtained by
slight deformations of $C_{j}$ and $C''_{j}$ in small neighborhoods of the nodes of ${\cal G}$ and
${\cal G}^e$, respectively, using the procedure described above in the context of $C'$. It is
easy to see that $\bar{C}_{j}$, $\bar{C}'_{j}$, and $\bar{C}''_{j}$ can be deformed to each other
within the immersion space, in particular $f(\bar{C}_{j})$, $f(\bar{C}'_{j})$, and
$f(\bar{C}''_{j})$ represent homologically equivalent cycles in $M^{3}$. Therefore, we have
\begin{eqnarray}
\label{top-to-comb}
q(\gamma)&=&q\left(\bigsqcup_{j=1}^{n}C_{j}\right)=\sum_{j=1}^{n}q(\bar{C}_{j})+\sum_{j<k}C''_{j}\cdot
C''_{k}=\sum_{j=1}^{n}s\left(f(\bar{C}_{j})\right)+\sum_{j=1}^{n}p\left(f(\bar{C}''_{j})\right)
+\sum_{j<k}C''_{j}\cdot C''_{k} \nonumber \\ &=&
\sum_{j=1}^{n}s\left(f(\bar{C}'_{j})\right)+n+\left(\sum_{j=1}^{n}i\left(f(\bar{C}''_{j})\right)
+\sum_{j<k}C''_{j}\cdot C''_{k}\right)=\sum_{j=1}^{n}s\left(f(C'_{j})\right)+n+N,
\end{eqnarray}
where $N=N(\gamma,{\bm\xi})$ is the total number of intersections and self-intersections. Combining
Eq.~(\ref{top-to-comb}) with the relation
$\varepsilon_{C}(C_{j})=\varepsilon_{T}\left(f(C'_{j})\right)=(-1)^{s\left(f(C'_{j})\right)}$ we
complete our topological interpretation of the statement of Lemma~\ref{Lemma}.

\subsection{Topology on Surface Graphs: related approaches and future work}

A topological interpretation of Lemma~\ref{Lemma} discussed in this Appendix is based on
establishing an equivalence between the definitions of $\varepsilon(C)$ in terms of a Kasteleyn
orientation on the extended graph Eq.~(\ref{define-epsilon}) and a spinor structure on
$\mathbb{R}^{2}$, respectively. This description allows generalization from $\mathbb{R}^{2}$ to surfaces $M^{2}$ with finite
genus. A relation between Kasteleyn orientations on surface graphs ${\cal G}\subset M^{2}$
and spinor structures on the surface has been established by Cimasoni and Reshetikhin \cite{07CR},
who demonstrated that a combination of a Kasteleyn orientation and a valid dimer configuration on
${\cal G}$ produces a spinor structure on $M^{2}$. We have implemented a very similar, still
different approach by associating a spinor structure on $\mathbb{R}^{2}$ with a Kasteleyn
orientation on the extended graph ${\cal G}^e\subset \mathbb{R}^{2}$. In a forthcoming
publication we will demonstrate that generalization of our construction solves the problem in the general $M^2$
surface case, thus allowing to relate partition function of a general discrete-variables Wick model
defined on a graph embedded in a surface of genus $g$ to a sum over $2^{2g}$ contributions, each
associated with a partition function of a fermion model of an allowed equivalence classes of spinor
structures of $M^2$.

It is also worth pointing out that Cimasoni and Reshetikhin \cite{07CR} used the aforementioned
self-intersection invariant $N(C)$ implicitly since the quadratic form $q$ they have used to study
the dimer model on surface graphs and the one considered here (in our planar case $q\equiv 0$) are
actually the same. However, the objects that naturally arise in the studies of the dimer model on
surface graphs are represented by simple unions of non-self-intersecting loops
\cite{99GLa,99GLb,07CR,08CR}, and, therefore, the issue of intersections is never raised in the
context of the dimer model. On the contrary, intersections play a key role in establishing the
correspondence between the Wick and fermion models, and should be handled explicitly. To the best
of our knowledge, in the context of $2d$ statistical mechanics for the first time the
self-intersection invariant, of the type discussed above, however in a gauge of another form, has been
brought up by Kac and Ward \cite{52KW} to come up with an exact solution for the two-dimensional
Ising model on a square grid. In the later work, where the Majorana spinors on irregular lattices
have been studied \cite{86BM,01BBJK}, the authors have been using the complex valued phase factors $e^{i\theta}$ to establish relations
between the free fermion and binary (Ising) models. This can be viewed as an extension of the
approach of \cite{52KW} from a regular to irregular lattices. In this manuscript we have, in a
sense, followed \cite{07CR} by working with the ``Kasteleyn" phase factors $\pm 1$.
The equivalence between the approaches that use different phase factors can be
established using a gauge transformation of the field ${\bm A}$ that describes a spinor structure
according to Eq.~(\ref{define-Dirac-phase}).

\section{Extending the gauge group}
\label{sec:gauge-extend}

In this Appendix we discuss an alternative way of reducing the three illustrative models, dimer,
ice, and Ising discussed in Section \ref{sec:gauge}, to the WBG model described by
Eqs.~(\ref{Z_WBG},\ref{W-expressions}). In Subsection \ref{subsec:gauge-reminder} the dimer model
has been converted to the WBG model via an example of the gauge transformation described in
Eqs.~(\ref{Gauge1},\ref{skew},\ref{Z_EBG}). The ice and Ising models strictly speaking do not
belong to the class of vertex models, considered in this manuscript, and an additional geometrical
transformation (an extension of the original graph) has been introduced in Sections
\ref{subsec:ice},\ref{subsec:Ising}.

In this Appendix we demonstrate that the equivalence of the ice model and the Ising model to respective
WBG models can be established without extending the original graph. To that end we extend the gauge
group (i.e., the group of gauge transformations) by relaxing the constraints given by
Eq.~(\ref{skew}). Specifically, we relax the constraint $\pi_{ab}=\pi_{ba}$ in the definition of
the edge variables, so that the edge variables ${\bm\pi}_{\alpha}$ are represented by pairs
${\bm\pi}_{\{a,b\}}=\left(\pi_{ab},\pi_{ba}\right)$ of binary variables. In addition to the vertex
weights $f_{a}({\bm\pi}_{a})$ we further introduce the edge weights
$g_{\alpha}({\bm\pi}_{\alpha})$. We will also use a notation $g_{ab}(\pi_{ab},\pi_{ba})$ with a
natural constraint $g_{ab}(\pi_{ab},\pi_{ba})=g_{ba}(\pi_{ba},\pi_{ab})$ that makes this notation
consistent with the original one. The partition function of such an Extended EBG, thus EEBG, becomes
\begin{eqnarray}
\label{define-Z-ext} Z_{EEBG}=\sum_{{\bm\pi}}\prod_{a\in{\cal
G}_{0}}f_{a}({\bm\pi}_{a})\prod_{\alpha\in{\cal G}_{1}}g_{\alpha}({\bm\pi}_{\alpha})
\end{eqnarray}
A gauge transformation of an extended model is described by a set
$\left\{G_{ab}(\pi_{ab},\pi'_{ab})\right\}_{\{a,b\}\in{\cal G}_{1}}$ of local invertible $2\times
2$ matrices with no other restrictions. A gauge transformation for the vertex functions is still
given by Eq.~(\ref{Gauge1}), whereas the edge functions are transformed according to the rule
\begin{eqnarray}
\label{Gauge-edge-1} g_{\{a,b\}}({\bm\pi}_{\{a,b\}})\to
\tilde{g}_{\{a,b\}}({\bm\pi}_{\{a,b\}})=\sum_{{\bm\pi}'_{\{a,b\}}}G^{ab}(\pi_{ab},\pi'_{ab})
G^{ba}(\pi_{ba},\pi'_{ba})g_{ab}(\pi'_{ab},\pi'_{ba})
\end{eqnarray}
where $G^{ab}$ denote the matrices inverse to $G_{ab}$, i.e.,
\begin{eqnarray}
\label{define-G-inverse}
\sum_{\pi_{ab}}G^{ab}(\pi''_{ab},\pi_{ab})G_{ab}(\pi_{ab},\pi'_{ab})=\delta(\pi''_{ab},\pi'_{ab})
\end{eqnarray}
Obviously the partition function of the extended model is invariant with respect to the extended
gauge transformations.

We further note that the EBG model (\ref{Z_EBG}) is a particular case of the EEBG
model (\ref{define-Z-ext}) that corresponds to the choice of the edge functions in a form
\begin{eqnarray}
\label{g-as-delta} g_{ab}(\pi_{ab},\pi_{ba})=\delta(\pi_{ab},\pi_{ba})
\end{eqnarray}
and the EBG-model gauge transformations, i.e., the ones that satisfy the constraints of
Eq.~(\ref{skew}), can be viewed as the extended gauge transformation that preserve the edge
functions of the form, given by Eq.~(\ref{g-as-delta}). Obviously, any extended model is equivalent
to an EBG model.

The ice model represents a class of graph-based models, hereafter referred to as arrow models on
graphs (AG), where a configuration is given by a graph orientation, i.e, a set of arrows associated
with the non-oriented edges, rather than binary variables residing on the edges. Any AG model can
be viewed as a particular case of EEBG model by associating with a local edge configuration of an
AG model on $\{a,b\}$, determined by an arrow $a\to b$, the edge configuration ${\bm\pi}_{ab}$ with
$\pi_{ab}=0$ and $\pi_{ba}=1$. Obviously the obtained EEBG model corresponds to the choice of the
edge functions represented by the same off-diagonal matrix with $g_{ab}(0,0)=g_{ab}(1,1)=0$ and
$g_{ab}(0,1)=g_{ab}(1,0)=1$. A homogeneous gauge transformation of this form is
\begin{eqnarray}
\label{arrow-to-EBG} G^{ab}=\frac{1}{\sqrt{2}}\left(\begin{array}{cc} 1 & 1\\ i &
-i\end{array}\right)
\end{eqnarray}
transforms the edge functions $g_{\alpha}$ of an AG model into respective ones described by
Eq.~(\ref{g-as-delta}). This establishes an equivalence between an AG model and respective EBG
model. Note that the edge function for an AG model is symmetric and, therefore, can be viewed as a
quadratic form. Then the gauge transformation can be interpreted as the quadratic form
diagonalization. If performed over real numbers the result has a diagonal form with the signature
$(1,-1)$. The imaginary units in the matrix elements in Eq.~(\ref{arrow-to-EBG}) are responsible
for transforming the obtained diagonal matrix into the form given by Eq.~(\ref{g-as-delta}).

Consider a subclass of AG models, hereafter referred to as even AG models, whose vertex functions
stay invariant upon a change the direction of all arrows, associated with the vertex. (Note that the
ice model belongs to this subclass.) It is easy to show that the described above homogeneous gauge
transformation transforms an even AG model into an even EBG model. Also note that in this even case
the vertex functions of the resulting even EBG model are real, despite of the presence of the
imaginary unit in the gauge transformation given by Eq.~(\ref{arrow-to-EBG}). The suggested gauge
transformation identifies the subclass of AG models (including, e.g., the ice model) resulting in
WBG models upon the gauge transformation given by Eq.~({arrow-to-EBG}), and, therefore are easy.

The Ising model on an arbitrary graph ${\cal G}$ can be also viewed as an EEBG model  with the
vertex and edge function given by
\begin{eqnarray}
&& f_a({\bm\pi})=\left\{\begin{array}{cc} 1, & \forall b,c\in\delta_{\cal G}(a)\quad \pi_{ab}=\pi_{ac},\\
0, & \mbox{otherwise}\end{array}\right.,\label{Is2-gauge}\\
&& g_{ab}(\pi_{ab},\pi_{ba})=\left\{\begin{array}{cc}
\gamma, & \pi_{ab}=\pi_{ba}\\
\mu, & \pi_{ab}\neq\pi_{ba}\end{array}\right., \label{Is3-gauge}
\end{eqnarray}
which differ from Eqs.~(\ref{Is2}) and (\ref{Is3}) actually by notations only. The gauge
transformation that diagonalizes the vertex functions $g_{\alpha}$, viewed as quadratic forms,
turns the Ising model into an even EBG model that belongs to the WBG subclass, as outlined in
subsection \ref{subsec:Ising}.

Finally, we note that the approach discussed in this Appendix, based on the extended
gauge transformation, can be viewed as an alternative formulation leading to the well-known
high-temperature expansion of the Ising model, when the expansion is interpreted as a loop series
with only $\mathbb{Z}_2$-cycles providing non-zero contributions.

\end{document}